\documentclass[journal]{IEEEtran}

\ifCLASSOPTIONcompsoc 
\usepackage[caption=false,font=normalsize,labelfont=sf,textfont=sf]{subfig}
\else
\usepackage[caption=false,font=footnotesize]{subfig}
\fi 

\usepackage{cite}
\usepackage{amssymb}
\usepackage{mathrsfs}
\usepackage{amsmath}
\usepackage{graphicx}
\usepackage{array}
\usepackage{multirow}
\usepackage{color} 
\usepackage{xcolor}
\usepackage{epstopdf}
\usepackage{cite}
\usepackage{amsfonts} 

\usepackage{diagbox} 

\usepackage[twocolumn,letterpaper]{geometry}

\newtheorem{theorem}{\bf{Theorem}}

\newtheorem{condition}{\bf{Assumption}}

\newtheorem{definition}{\bf{Definition}}
\newtheorem{example}{\bf{Example}}

\newtheorem{lemma}{\bf{Lemma}}

\newtheorem{problem}{\bf{Problem}}
\newtheorem{proposition}{\bf{Proposition}}
\newtheorem{remark}{\bf{Remark}}

\begin{document} 
	 
\title{\Large Prescribed-Time Fully Distributed Nash Equilibrium Seeking in Noncooperative Games}
 
\author
{Zhi Feng and Guoqiang Hu 


\thanks{
This work was supported in part by Singapore Ministry of Education Academic Research Fund Tier 1 RG180/17 (2017-T1-002-158) and in part by Singapore Economic Development Board under EIRP grant S14-1172-NRF EIRP-IHL. 

Z. Feng and G. Hu are with the School of Electrical and Electronic Engineering, Nanyang Technological University, Singapore 639798 (E-mail: zhifeng@ntu.edu.sg; gqhu@ntu.edu.sg). 
}  
} 
 
 
\maketitle 
 
\begin{abstract}
In this paper, we investigate a prescribed-time and fully distributed Nash Equilibrium (NE) seeking problem for continuous-time noncooperative games. By exploiting  pseudo-gradient play and consensus-based schemes, various distributed NE seeking algorithms are presented over either fixed or switching communication topologies so that the convergence to the NE is reached in a prescribed time. In particular, a prescribed-time distributed NE seeking algorithm is firstly developed under a fixed graph to find the NE in a prior-given and user-defined time, provided that a static controller gain can be selected based on certain global information such as the algebraic connectivity of the communication graph and both the Lipschitz and monotone constants of the pseudo-gradient associated with players' objective functions. Secondly, a prescribed-time and fully distributed NE seeking algorithm is proposed to remove global information by designing heterogeneous dynamic gains that turn on-line the weights of the communication topology. Further, we extend this algorithm to accommodate jointly switching topologies. It is theoretically proved that the global convergence of those proposed algorithms to the NE is rigorously guaranteed in a prescribed time based on a time function transformation approach. In the last, numerical simulation results are presented to verify the effectiveness of the designs. 

\vspace*{2pt}       
The advantages of the proposed NE seeking algorithm include: 1) the convergence time is user-defined according to task requirements, which is neither dependent on any initial states or the parameters of the algorithm; 2) the proposed algorithm is fully distributed without requiring any global information on the graph's algebraic connectivity, the pseudo-gradient's Lipschitz and monotone constants, and the number of players; and 3) the communication graph is allowed to be jointly switching. The aforementioned requirements can improve the practical relevance of the problem to be addressed and meanwhile, it poses some technical challenges to the algorithm design and stability analysis, which makes that the NE seeking algorithms in the existing literature cannot be directly applied. 
\end{abstract}

\begin{IEEEkeywords}
Noncooperative game, NE seeking algorithm, Fully distributed, Prescribed-time convergence, Switching topologies. 
\end{IEEEkeywords}

\IEEEpeerreviewmaketitle

\vspace*{-6pt}
\section{Introduction}
Distributed Nash equilibrium (NE) seeking of non-cooperative games has become a hot research topic during the past decade  due to its broad applications in multi-robot systems \cite{Basar87AT}, mobile sensor networks \cite{Johansson12TAC}, smart grids \cite{Ye17Tcyber}, and so on. In contrast to early works (e.g., \cite{Shamma05TAC,Basar12TAC,Scutari04TIT}) with a complete information setting, players in distributed NE seeking have limited local information, i.e., each player needs to make the decision based on the local or relative information, e.g., information from its neighbors, to optimize its own cost function. The main challenges of distributed NE seeking exist in twofold: 1) each player's objective function is dependent on the other players' actions and therefore, its strategy is directly influenced by other players; and 2) each player is required to not only update its own strategy, but also to communicate on networks to estimate other players' actions.   
 
\textit{Related literature:} gradient-based NE seeking algorithms with average consensus or leader-follower  consensus designs are popular techniques to find the NE of networked games. 
Distributed NE seeking issues of continuous-time games have been addressed in  \cite{Ye17TAC,Ye19TAC,Pavel19TAC,Liang17AT,Deng19TNNLS,Pavel19AT,Pang20TAC,Wang19ICCA,Persis19AT}. In particular, distributed algorithms are proposed in \cite{Ye17TAC} and \cite{Ye19TAC} by combining the leader-follower consensus designs and gradient-play strategies over an undirected and connected graph. The singular perturbation method is employed to ensure that the consensus design can be faster than the gradient updating part, and the semi-global asymptotic convergence is thus achieved. 
The authors in \cite{Pavel19TAC} exploit some incremental passivity properties of pseudo-gradients to illustrate that the estimates of the proposed augmented gradient dynamics converge to the NE exponentially under graph coupling conditions. The singular perturbation design is further developed to provide an adjustable singular perturbation parameter so as to relax graph conditions. Distributed NE seeking of aggregative games is investigated in \cite{Liang17AT} and \cite{Deng19TNNLS}, where 
the former presents a nonsmooth NE seeking algorithm with identical constant gains, 
while the latter designs the controller based on a singular perturbation parameter.  
An alternating direction method of multipliers' design with the constant step-size is given in \cite{Pavel19AT}. 
Recently, \cite{Pang20TAC} uses a gradient-free NE seeking scheme for limited cost knowledge, 
where an almost sure NE convergence is obtained by a diminishing step-size, while a uniformly ultimately bounded (UUB) convergence is achieved by a fixed step-size.     

\textit{Research gap:} to the best of knowledge, there are some design limitations that have not been dealt with in the aforementioned works in \cite{Ye17TAC,Ye19TAC,Pavel19TAC,Liang17AT,Deng19TNNLS,Pavel19AT,Pang20TAC,Wang19ICCA,Persis19AT}. Firstly, one observation is that those aforementioned NE seeking results can guarantee a semi-global asymptotic convergence in \cite{Ye17TAC,Ye19TAC}, an asymptotic or exponential convergence in \cite{Pavel19TAC,Liang17AT,Deng19TNNLS}, a linear convergence in \cite{Pavel19AT}, or a UUB convergence in \cite{Pang20TAC}. That is, the NE is only seek over an infinite-time horizon. Besides, the convergence rate heavily relies on the players' initial conditions, the communication topology structure, and the control parameter of algorithms, which makes it hard to off-line preassign convergence time. \textbf{\textit{Hence, it is desirable to propose a prescribed-time distributed NE seeking algorithm with the convergence time prior-given and user-defined according to game requirements.}} Secondly, another observation is that all works in \cite{Ye17TAC,Ye19TAC,Pavel19TAC,Liang17AT,Deng19TNNLS,Pavel19AT,Pang20TAC,Wang19ICCA,Persis19AT} require the static control gains depending on global information such as the algebraic connectivity of graphs, the Lipschitz and monotone constants of pseudo-gradients, and the number of players. Notice that in practice, it is often hard to verify those global information in a larger-scale multi-agent system. 
In addition, the used singular 
perturbation control gain
in \cite{Ye17TAC,Ye19TAC,Pavel19TAC} has to be 
high-enough, which might be difficult to estimate and  implement. Moreover, the non-smooth algorithm 
in \cite{Liang17AT} based on a signum function often brings undesirable chattering behaviors. \textbf{\textit{Hence, it is desirable to develop a fully distributed NE seeking algorithm without requiring any global information, which is smooth with heterogeneous control gains that can turn on-line the control effort.}}    

This paper focuses on the prescribed-time and fully distributed research of NE seeking in noncooperative games, considering that very few of the existing literature has investigated these properties in distributed Nash games. The study of the NE convergence rate is partially inspired by the recent finite-/fixed-time research in the distributed consensus and optimization (e.g., \cite{Zhao16AT,Guan12TCS,SongTcyber18,ZhaoAT18, Zan19TAC,Pilloni16VSC,Ren14ACC,Chen18AT,Hu19TCNS,Ding19ACC}), which suffers from certain design limitations that make them not suitable for prescribed-time distributed NE games  (comparison details are summarized in Remark 6, and  omitted here). 

In this paper, we provide some feasible and easy-implemented algorithms to accomplish the task of distributed NE seeking with the arbitrary convergence time guarantee. The main technique is to exploit a time transformation function method, under which we  transform the proposed smooth prescribed-time algorithm into an infinite-time interval. Then, the Laypunov stability theory is still allowed to analyze the convergence of the NE in this infinite-time interval. As compared to existing distributed NE seeking works, the main contributions of this work can be summarized: 

\begin{itemize}
\item To the best of our knowledge, this paper is the first work to present a smooth and prescribed-time distributed NE seeking architecture to solve this issue. Different from the existing distributed NE results in \cite{Ye17TAC,Ye19TAC,Pavel19TAC,Liang17AT,Deng19TNNLS,Pavel19AT,Pang20TAC,Wang19ICCA,Persis19AT}, the salient feature of the proposed algorithm is that the arbitrary convergence time for reaching the consensus of all players' estimates and seeking the NE globally, is independent of any initial conditions and design parameters, thus can be explicitly pre-specified. This fast convergence is of great significance for a wide range of NE game applications in large-scale systems.   

\item Moreover, a prescribed-time and fully distributed NE seeking algorithm is developed, where the design not only provides a fast convergence but also not rely on any global information like the graph's algebraic connectivity, the pseudo-gradient's Lipschitz and monotone constants, and the number of players as required in \cite{Ye17TAC,Ye19TAC,Pavel19TAC,Liang17AT,Deng19TNNLS,Pavel19AT,Pang20TAC,Wang19ICCA,Persis19AT}. The global fast  convergence to the NE is achieved through adaptively adjusting a dynamic gain on the edges of the communication graph. In the absence of prescribed-time requirements, the  works in \cite{Wang19ICCA} and \cite{Persis19AT} also aim to solve distributed NE seeking issues via the dynamic gain. Unfortunately, the former adopting a decaying control gain requires global information, and the latter comes at the cost of two-hop communication among players.     

\item  Based on the time transformation method and the LaSalle's invariance principle, it is shown that the proposed prescribed-time and fully distributed architecture guarantees that the NE is globally stable by mild and standard assumptions on the players' pseudo-gradients and communication graphs. Lastly, the extension is presented to accommodate jointly switching topologies. In addition, the proposed algorithms are suitable for a general Nash game where the cost functions are coupled with other agents, which can cover the potential game in \cite{Ye17Tcyber} and aggregative games in 
\cite{Liang17AT,Deng19TNNLS,Pavel19AT}.
\end{itemize}

\vspace{1pt}
The paper is organized as follows. Section II gave  mathematical preliminaries. The non-cooperative game and the main objective are presented in Section III. Some prescribed-time distributed algorithms are proposed in Section IV with convergence analysis. Examples and numerical simulation results are given in Section V, followed by the conclusion in Section VI.

\section{Preliminaries} 
\vspace*{-8pt} 
\subsection{Notation \label{Notation}} 
Denote $\mathbb{R}$, $\mathbb{R}^{n}$, and $\mathbb{R}^{n\times m}$ as the sets of the real numbers, real $n$-dimensional vectors and real $n\times m$ matrices, respectively. 
Let $0_{n}$ ($1_{n}$) be the $n\times 1$ vector with all zeros (ones) and $I_{n}$ be the identity matrix. Let col$(x_{1},...,x_{n})$ and diag$\{a_{1},...,a_{n}\}$ be a column vector with entries $x_{i}$ and a diagonal matrix with entries $a_{i}$, $i=1,\cdots ,n$, respectively. The symbols $\otimes $ and $\left\Vert \cdot \right\Vert $ represent the Kronecker product and the Euclidean norm, respectively. Given a real symmetric matrix $M$, let $ M > 0 $ ($ M \geq 0 $) denote that $ M $ is positive (or positive semi-definite), and $\lambda _{\min }(M)$, $\lambda _{\max }(M)$ are its minimum and maximum eigenvalues, respectively. 

\vspace*{-8pt}
\subsection{Convex analysis \label{ConvexAnalysis}}
A function $ f: \mathbb{R}^{n} \rightarrow  \mathbb{R} $ is convex if $ f(ax+(1-a)y) \leq af(x)+(1-a)f(y)$ for any scalar $a \in [0,1] $ and vectors $ x,y \in \mathbb{R}^{n}$. $ f $ is locally Lipschitz on $ \mathbb{R}^{n}$ if it is locally Lipschitz at $ x $ for $\forall x \in \mathbb{R}^{n} $. If $ f $ is a differentiable function, $ \triangledown f $ denotes the gradient of $ f $. 
A vector-valued function (or mapping) $ F: \mathbb{R}^{n} \rightarrow \mathbb{R}^{n}$ is said to be $ \iota_{F} $-Lipschitz continuous if, for any $ x,y \in \mathbb{R}^{n}$, $ \| F(x)-F(y)\| \leq \iota_{F} \| x-y \| $. Function $ F: \mathbb{R}^{n} \rightarrow \mathbb{R}^{n}$ is (strictly) monotone if, for any $ x,y \in \mathbb{R}^{n}$, $ (x-y)^{T} (F(x)-F(y)) (>) \geq 0 $. Further, $ F $ is a $ \varepsilon $-strongly monotone, if for any scalar $ \varepsilon>0 $, and $ x,y \in \mathbb{R}^{n}$, $ (x-y)^{T} (F(x)-F(y)) > \varepsilon \|x-y\|^{2} $.   For a function $f$, it is said to be $\mathcal{C}^{m}$ if it is $ m $th continuously differentiable.

\vspace*{-8pt}
\subsection{Graph Theory\label{Graph theory}} 
\textit{Fixed graph}: let $\mathcal{G}$ $=$ $\left\{ \mathcal{V},\mathcal{E}  \right\} $ represent a communication graph, where $\mathcal{V} \in \left\{ 1,2,...,N\right\} $ is a set of nodes and  $\mathcal{E} \subseteq  \mathcal{V\times V}$ is a set of edges. 
In this paper, we assume that there is no self loops in the graph, that is, $(i,i)\notin \mathcal{E}$. A path from node $ i_{1} $ to node $ i_{q} $ is a sequence of ordered edges in the form of $ (i_{1}, i_{2}), \cdots, (i_{q-1},i_{q}) $. A graph is said to be connected if there exists a path connecting each pair of distinct nodes. An edge $(i,j)\in \mathcal{E}$ denotes that $ i $th agent receives the information from $ j $th agent, but not vice versa. Graph $\mathcal{G}$\ is said to be undirected if for any 
$(i,j)\in \mathcal{E}$, 
$%
(j,i)\in \mathcal{E}$. Let $\mathcal{N}_{i}(\mathcal{G})$ $=$ $\left\{ j\in \mathcal{V\mid }(j,i)\in \mathcal{E}%
\right\} $ denote the neighborhood set of node $i$. The adjacency matrix of $\mathcal{G}$ is denoted as $\mathcal{A}$ $=$ $\left[ a_{ij}\right] $ $\in $ $ \mathbb{R}^{N\times N}$, where $a_{ij}>0$ if and only if $(j,i)\in \mathcal{E}$, else $a_{ij}=0$. The Laplacian matrix of $\mathcal{G}$ is denoted by $\mathcal{L}$ $\mathcal{=}$ $\mathcal{[}l_{ij} \mathcal{]} \in  \mathbb{R}^{N\times N}$, where $l_{ii}$ $=$ $\sum\nolimits_{j=1}^{N}a_{ij}$ and $%
l_{ij} $ $=$ $-a_{ij}$ if $i\neq j.$ Also,  $\mathcal{L}=D-\mathcal{A}$ with a matrix 
$ D= \text{diag} \{\sum_{j=1}^{N} a_{ij} \} \in  \mathbb{R}^{N\times N}$.

\textit{Switching graph}: let $\mathcal{G}^{\sigma(t)}= \left\{ \mathcal{V},\mathcal{E}^{\sigma(t)} \right\} $ be a switching graph, where $\mathcal{E}^{\sigma(t)}  \subseteq  \mathcal{V\times V} $ for all $ t \geq 0 $. Here, we call a time function $ \sigma (t): [0, \infty) \rightarrow \mathcal{P}=\{1,\cdots,m\} $ a piecewise constant switching signal if there exists a time sequence $ t_{0}=0<t_{1}<t_{2}<\cdots $ with $ t_{k+1}-t_{k} \geq \tau>0 $ for certain  dwell time $ \tau$  and $ k\geq 0 $ so that during $ [t_{k},t_{k+1})$, $ \sigma (t)=i $ for some $ i \in \mathcal{P} $ and this graph $ \mathcal{G}_{i} $ is time-invariant. For convenience of analysis, $ \mathcal{L}^{\sigma(t)} $ represents the Laplacian matrix of the undirected graph $ \mathcal{G}^{\sigma(t)} $. 


\vspace*{1pt}
\begin{condition} \label{AssumptionGraph}
The fixed undirected graph $ \mathcal{G}$ is connected. 	
\end{condition}


\vspace*{1pt}
\begin{condition} \label{AssumptionGraphSwitching}
For the graph $ \mathcal{G}^{\sigma(t)} $,  there exists a subsequence $ \{ l_{k}\} $ of $ \{l: l=0,1,\cdots \} $ with $ t_{l_{k+1}}-t_{l_{k}} <\nu $ for some $ \nu>0 $ so that the union graph $ \cup^{l_{k+1}-1}_{j=l_{k}} \mathcal{G}_{\sigma(t_{j})}$ is connected.  	
\end{condition}   
   
\begin{remark}
Assumption \ref{AssumptionGraph} is standard in distributed NE seeking, while Assumption \ref{AssumptionGraphSwitching} is called jointly connected (e.g., see \cite{Huang12TAC} and \cite{Cai16TAC}). The assumption \ref{AssumptionGraphSwitching} allows the communication topology to be disconnected at any time instant.
\end{remark}

\section{Problem Formulation}
\vspace*{-2pt} 
\subsection{Non-cooperative Game over Networks} 
\vspace*{-2pt}  
In this paper, we consider a multi-agent network consisting of $N$ players, which form a N-player non-cooperative game defined as follows. For each agent $ i \in \mathcal{V}$, the $ i $th player aims to minimize its cost function $ J_{i}(x_{i},x_{-i}): \mathbb{R}^{n_{i}} \rightarrow \mathbb{R} $ by choosing its strategy $ x_{i} \in \mathbb{R}^{n_{i}} $, and $ x_{-i}=\text{col}(x_{1},\cdots,x_{i-1},\cdots,x_{N}) $ is the strategy profile of the whole strategy profile except for player $ i $. Let $ x= \\ (x_{i},x_{-i}) $ represent all players' action profile. Alternatively, let $ x \\ =  \text{col}(x_{1},\cdots,x_{N}) \in \mathbb{R}^{n} $, $ n=\sum_{i\in \mathcal{V}}n_{i} $.   


\begin{definition} \label{NEDefinition} \textit{(Nash equilibrium)}
A strategy profile $ x^{*}=(x^{*}_{i}, \\ x^{*}_{-i}) \in \mathbb{R}^{n}$ is said to be an Nash equilibrium of the game if 
\vspace*{-3pt}
\begin{equation}
J_{i}(x^{*}_{i},x^{*}_{-i}) \leq J_{i}(x_{i},x^{*}_{-i}), \ \text{for} \ \forall   x_{i} \in \mathbb{R}^{n_{i}}, \ i\in \mathcal{V}. \label{NECondition} 
\end{equation}
\end{definition} 

Condition (\ref{NECondition}) means that all players simultaneously take their own best (feasible) responses at the NE $ x^{*}$, where no player can unilaterally decrease its cost by changing its strategy. 


\begin{condition} \label{AssumptionConvex}
For each player $ i $, $ J_{i}(x_{i},x_{-i}) $ is $ \mathcal{C}^{2} $, strictly convex, and radially unbounded in $ x_{i} $ for each $ x_{-i} $. 
\end{condition}

Under Assumption \ref{AssumptionConvex}, it follows from \cite{Başar95} that an NE $ x^{*}$ exists, and satisfies $  \triangledown_{i} J_{i}(x^{*}_{i},x^{*}_{-i})=0_{n_{i}}$, and $ \triangledown_{i} J_{i}(x_{i},x_{-i})=\partial J_{i}(x_{i}, \\ x_{-i})/  \partial x_{i} \in \mathbb{R}^{n_{i}}$ represents the partial gradient of player $ i $'s cost with respect to its own action $ x_{i} $. We define 
\vspace*{-3pt}
\begin{equation}
F(x) \triangleq \text{col}( \triangledown_{1} J_{1}(x_{1},x_{-1}), \cdots,  \triangledown_{N} J_{N} (x_{N} \\ , x_{-N})), 
\end{equation} 
where $ F(x) \in \mathbb{R}^{n}$ denotes the pseudo-gradient (the stacked vector of all players' partial gradient). Thus, we have $ F(x^{*})=0_{n} $.

\begin{condition} \label{AssumptionSmooth}
The pesudogradient $ F $ is $ \varepsilon $-strongly monotone and $ \iota_{F} $-Lipschitz continuous for certain  constants $ \varepsilon, \iota_{F} >0 $.  
\end{condition}

\begin{remark}
Assumptions \ref{AssumptionConvex} and \ref{AssumptionSmooth} were widely used in existing works (e.g., \cite{Pavel19TAC,Liang17AT,Deng19TNNLS,Wang19ICCA,Pavel19AT,Pang20TAC}) to guarantee the unique NE $ x^{*} $. 
\end{remark}

\vspace*{-5pt} 
\subsection{Main Objective} 
This work aims to address a prescribed-time and fully distributed NE seeking problem of noncooperative games as follows.

\vspace*{2pt}
\begin{problem} \label{Problem}
(\textbf{Distributed NE Seeking in Prescribed-Time}) \\
Consider a non-cooperative game consisting of $ N $ players communicating over a communication network. Design a NE seeking algorithm such that all players can exactly reach the NE $ x^{*} $ with prescribed-time and fully distributed convergence features. 
\vspace*{-3pt}
\begin{equation} \label{P} 
\begin{split} 
& \text{minimize}   
\  J_{i}(x_{i}(t),x_{-i}(t)), \ x_{i}(t)\in \mathbb{R}^{n_{i}}, \ i\in \mathcal{V}, \\
&\text{subject to:} \ 
\dot{x}_{i}(t)=u_{i}(t),  \ \forall  t \in [0,T), \ T>0. 
\end{split}
\end{equation} 
\end{problem}

\begin{remark}
In contrast to existing works in \cite{Ye17TAC,Ye19TAC,Pavel19TAC,Liang17AT,Deng19TNNLS,Pavel19AT,Pang20TAC,Wang19ICCA,Persis19AT}, solving Problem 1 is much more challenging at least from the following aspects: \textit{(1) Prescribed-time convergence:} different from NE seeking results that guarantee a semi-global exponential convergence in \cite{Ye17TAC,Ye19TAC}, asymptotic or exponential convergence in  \cite{Pavel19TAC,Liang17AT,Deng19TNNLS}, linear convergence in \cite{Pavel19AT},  and UUB convergence in \cite{Pang20TAC}, it is desirable to solve Problem 1 in the  prescribed time (priori-given and user-defined) that is independent of any initial states, communication graphs, and control gains. \textit{(2) Player communication network:} the topology is jointly switching rather than the static graphs in   \cite{Ye17TAC,Ye19TAC,Pavel19TAC,Liang17AT,Deng19TNNLS,Pavel19AT,Pang20TAC,Wang19ICCA,Persis19AT}. 
\textit{(3) Design requirement:} propose a prescribed-time fully distributed NE seeking algorithm that does not require any global graph information and the Lipschitz and monotone constants of the pseudo-gradient. Due to aforementioned challenges, existing NE seeking algorithms cannot be directly applied.
\end{remark}

\section{Prescribed-Time Distributed NE Seeking} 
In distributed NE seeking games, each player $ i $ has no access to the full information of all players' strategies. Then, each agent $ i $ shall estimate all other players' strategies. Inspired by \cite{Pavel19TAC}, let each player combine its gradient-play dynamics with an auxiliary dynamics, i.e., implement the following dynamics:  
\vspace{-1pt}
\begin{equation}
\left\{ 
\begin{array}{c}
\hspace{-0.1em}
\dot{x}^{i}_{i}(t)=u^{i}_{i}(t), \ u^{i}_{i}(t)=-\triangledown_{i} J_{i}(x^{i}_{i}(t),\textbf{x}^{i}_{-i}(t)) + e^{i}_{i}(t),      \\ 
\hspace{-1.5em}
\dot{x}^{i}_{j}(t) =u^{i}_{j}(t),  \ u^{i}_{j}(t)=e^{i}_{j}(t), \ \forall  \ i, j \in \mathcal{V}, \ j\neq i,
\end{array}
\right. 
\label{ConsensusDesign}
\end{equation}
where player $ i $ maintains an estimate vector ${\textbf{x}}^{i}=\text{col}(x^{i}_{1}, \cdots, x^{i}_{i}, \\ \cdots, x^{i}_{N}) $ in which $ x^{i}_{j} $ is player $ i $'s estimate of player $ j $'s action, $ x^{i}_{i}=x_{i} $ is the player $ i $'s actual action, $ \textbf{x}^{i}_{-i} $ is the player $ i $' estimate vector without its own action, $ u^{i}_{i} =u_{i}$ is the player $ i $'s actual input, $ u^{i}_{j} $ is the other players' input, and $ e^{i}_{i}, e^{i}_{j} $ are to be developed. 
In (\ref{ConsensusDesign}), each player $ i $ updates $ x^{i}_{i} $ to reduce its own cost function and updates $ x^{i}_{j} $ to reach  consensus with the other players. In addition, each player $ i $ relies on its local estimated action $ \textbf{x}^{i}_{-i} $. 

For each player $ i $, (\ref{ConsensusDesign}) can be rewritten in a compact form
\vspace{-2pt} 
\begin{equation} 	
\dot{\textbf{x}}^{i}(t)=\textbf{u}^{i}(t), \  \textbf{u}^{i}(t)= -\mathcal{R}^{T}_{i} \triangledown_{i} J_{i}(\textbf{x}^{i}(t))  + \textbf{e}^{i}(t), \  i \in \mathcal{V},   
\label{NECompactDynamics} 
\end{equation} 
where $ \textbf{u}^{i} $ is the control input, $ \textbf{e}^{i}= \text{col}(e^{i}_{1},\cdots, e^{i}_{i}, \cdots,e^{i}_{N}) \in \mathbb{R}^{n}$ is a relative estimated error to be designed, and $ \mathcal{R}_{i} \in \mathbb{R}^{n_{i} \times n}$ used to align the gradient to action components, is a  matrix given by 
\vspace{-2pt} 
\begin{equation} 	
\mathcal{R}_{i}=\left [0_{n_{i}\times n_{1}} \cdots 0_{n_{i}\times n_{i-1}} \ I_{n_{i}\times n_{i}} \ 0_{n_{i}\times n_{i+1}}  \cdots 0_{ n_{i}\times n_{N} } \right].  
\label{NERmatrix} 
\end{equation}

\vspace{-10pt} 
\subsection{Prescribed-Time Distributed NE Seeking Design}
 
Before presenting the algorithm, the following lemma on a time transformation function is introduced to facilitate the design. 
 
\begin{lemma} \label{TimeFunction} \cite{Book14TimeFunction}
Consider a dynamical system that is described by $ \dot{x}(t)=f(t,x(t)) $ with $ x(0)=x_{0} $. Let $ \xi(t) $ denote the solution to this system and $ T>0 $ is the prior-given and user-defined time. Then, there exists a time transformation function $ t=\lambda(s), s \in [0,\infty) $ with $ \lambda(s)$ satisfying certain conditions:  
\vspace{-3pt}
\begin{subequations}\label{TFcondition}	
\begin{align}
 \lambda(0)=0, \lambda^{'}(0)=T \ - (\text{continuous differentiable}) & \label{TFC1} \\
  s_{1}>s_{2}\geq 0 \Rightarrow \lambda(s_{1})>\lambda(s_{2}) \ - (\text{strictly increasing}) & \label{TFC2} \\
 \lim_{s\rightarrow \infty} \lambda(s)=T, \lim_{s\rightarrow \infty} \lambda^{'}(s)=0 \ -  (\text{convergence in s}) &  \label{TFC3}
\end{align}
\end{subequations} 
so that for $ \psi(s) \triangleq \xi(t) $, we obtain  
\vspace{-2pt}
\begin{equation}
\psi^{'}(s)=\lambda^{'}(s)f(\lambda(s),\psi(s)), \ \psi(\lambda^{-1}(0) )=x_{0}, 
\end{equation}
where $ \psi^{'}(s)=d\psi(s)/ds $, $ \lambda^{'}(s)=d\lambda(s)/ds $ and $\lim_{s\rightarrow \infty} \psi(s)= \lim_{t\rightarrow T} \xi(t)$ for any $ t \in [0,T), s\in [0,\infty) $.
\end{lemma}	 
 
\vspace{1pt} 
In this work, the objective is to propose a prescribed-time NE seeking algorithm so that all players' estimates reach a consensus and converge to the NE in a time $ T $, and afterwards, this NE can be maintained for $ t \geq T $. To achieve this goal, we can choose the following time transformation function satisfying (\ref{TFcondition})   
\vspace{-2pt}
\begin{equation}
t=\lambda(s) \triangleq T(1-e^{-s}), \label{TimeTransformFunction} 
\end{equation}  
which implies that when $ s\rightarrow \infty $, $ t $ approaches  $T $  as shown in Fig. \ref{TimeTransFunction}, and $ \lambda(s) $ is continuously differentiable and strictly increasing, which satisfies (\ref{TFC1})-(\ref{TFC2}). Further, it can be  verified that $ \lambda^{'}(s)=Te^{-s} $ satisfies (\ref{TFC3}). According to (\ref{TimeTransformFunction}), the original time interval $ t\in [0,T) $ can thus be transformed into a new infinite-time interval $ s \in [0,\infty) $. Consequently, the associated stability analysis will be transformed to focus on this new time variable $ s\in [0,\infty)  $.


\begin{figure}[t!] 
	\centering
	\hspace*{-0.5em}
	\begin{tabular}{cc}	
		
		\hspace*{-1.2em}		
		\subfloat [$ \lambda(s) $] 
		{\includegraphics[width=4.8cm,height=3.5cm]{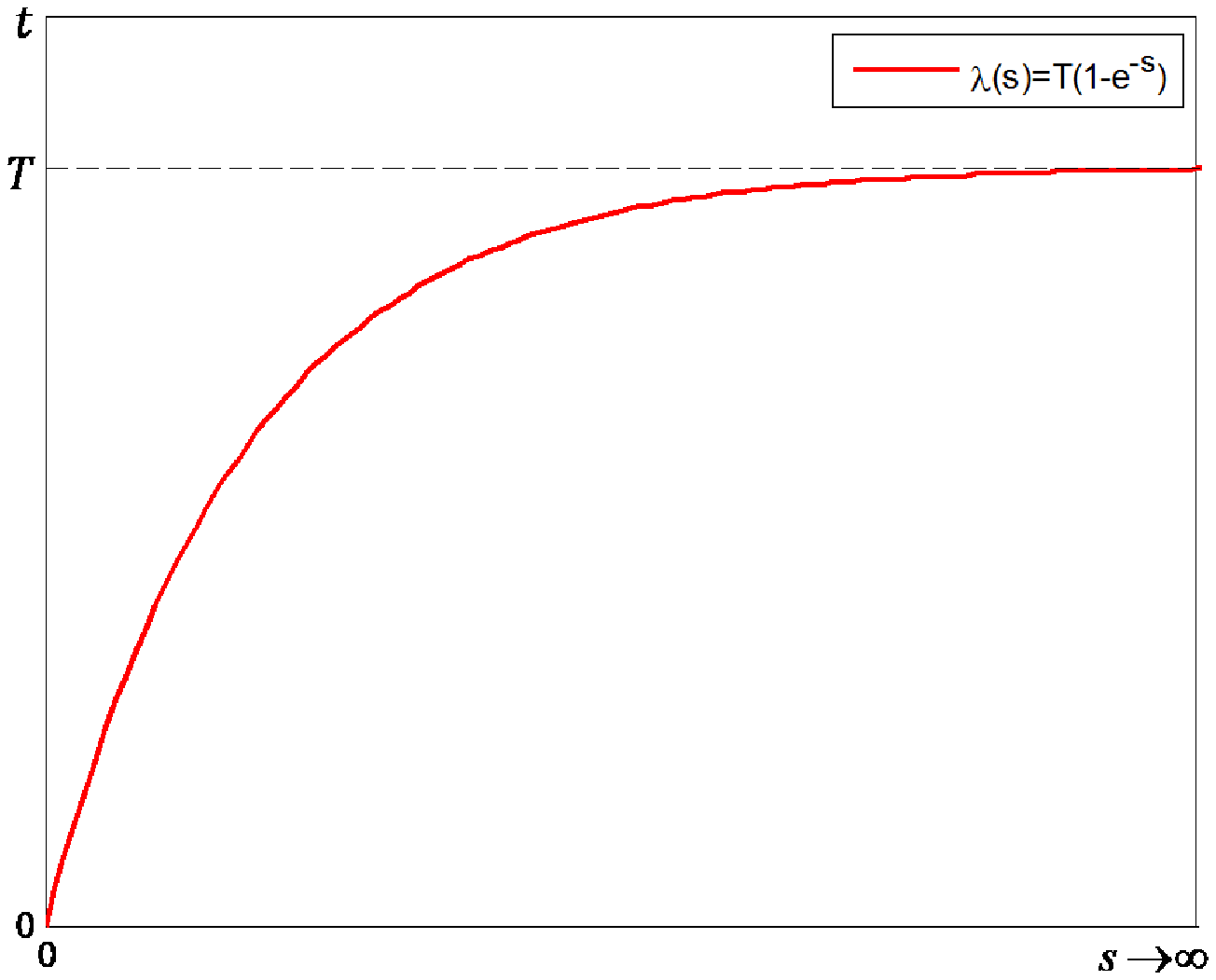}
			\label{T1}}
		
		\hspace*{-1.2em}		
		\subfloat [$ \lambda^{'}(s) $] 
		{\includegraphics[width=4.8cm,height=3.5cm]{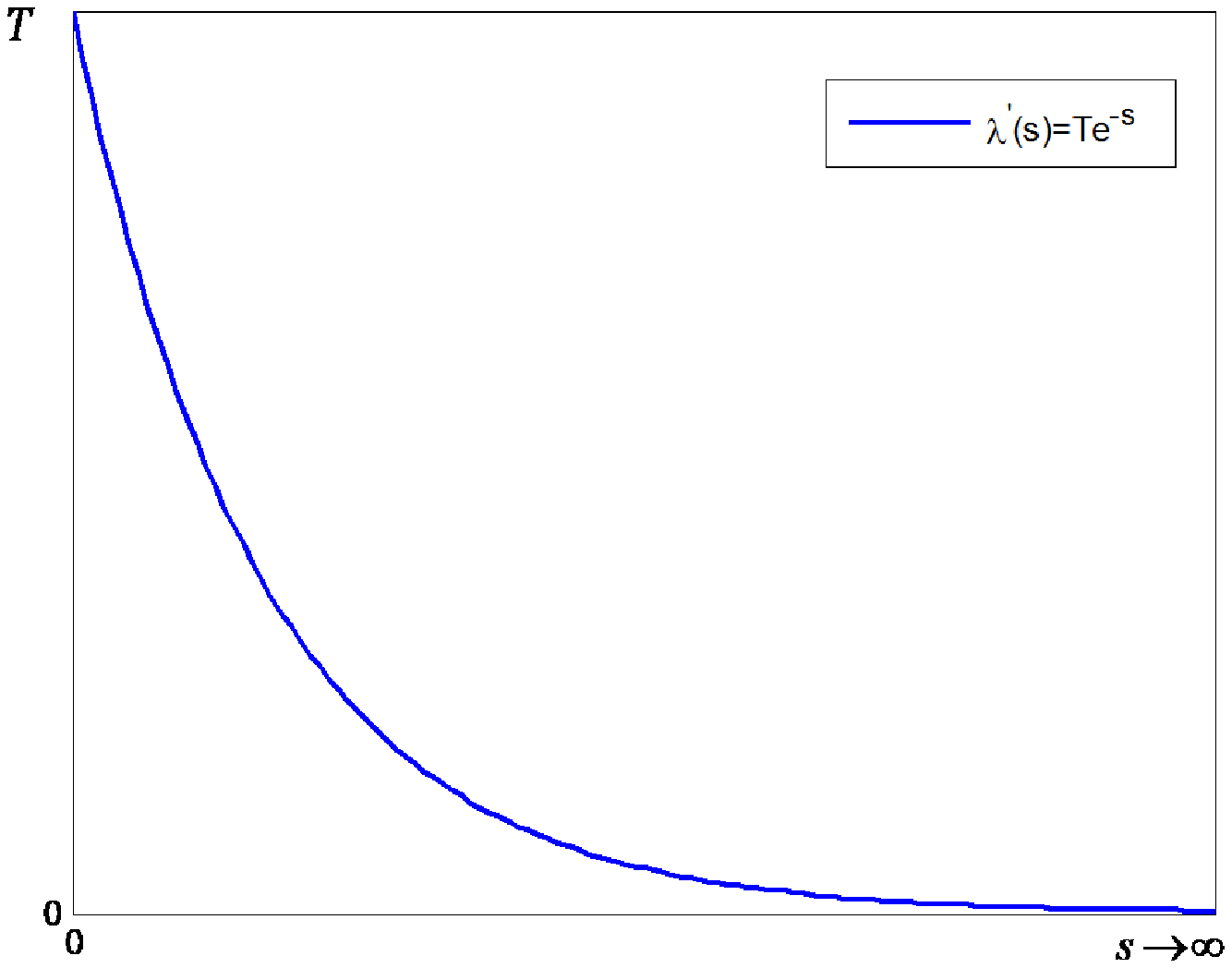}
			\label{T2}}
	\end{tabular}
	\vspace*{-3pt}
	\caption{The illustration of the time transformation function $ \lambda(s) $ and $ \lambda^{'}(s) $.}
	\label{TimeTransFunction}
\end{figure}

\textit{\textbf{Prescribed-time distributed NE seeking design}}: in light of (\ref{NECompactDynamics}), we present a new prescribed-time distributed NE seeking strategy so that the estimates of all players can not only reach a consensus, but also converge to the NE in a prior-given and user-defined $ T $, and this NE can be maintained for $ t \geq T $,  which is described by
\vspace{-3pt}
\begin{align} 
\dot{\textbf{x}}^{i}(t)&= (c+\frac{1}{T-t}) \textbf{u}^{i}(t), \  \textbf{u}^{i}(t)= -\mathcal{R}^{T}_{i} \triangledown_{i} J_{i}(\textbf{x}^{i}(t))  + \textbf{e}^{i}(t), \notag \\
\textbf{e}^{i}(t)&=-\kappa 
\sum\nolimits_{j = 1}^N  a_{ij} (\textbf{x}^{i}(t)-\textbf{x}^{j}(t)),  i, j  \in \mathcal{V},  t\in [0,T), 
\label{PrescribedTimeAlgorithm}
\end{align} 
where $c, \kappa>0$ are constant gains, and for $ t \geq T$, $\dot{\textbf{x}}^{i}(t)=c\textbf{u}^{i}(t) $.


\vspace{2pt} 
Next, denote the following stacked vectors and matrices 
\vspace{-3pt}
\begin{align}
\textbf{x} & =\text{col}(\textbf{x}^{1},\cdots,\textbf{x}^{N}), \  \mathcal{R}=\text{diag} \{ \mathcal{R}_{1}, \cdots,\mathcal{R}_{N} \},  \label{StackedVariable} \\
\textbf{e}&=\text{col}(\textbf{e}^{1},\cdots,\textbf{e}^{N}), \ \textbf{F(\textbf{x})} =\text{col}(\triangledown_{1} J_{1}(\textbf{x}^{1}),\cdots, \triangledown_{N} J_{N}(\textbf{x}^{N})). \notag 
\end{align}

Then, combining (\ref{NECompactDynamics}), (\ref{PrescribedTimeAlgorithm}) and (\ref{StackedVariable}) gives rise to the following closed-loop system in the sense of a compact form  
\vspace{-3pt}
\begin{equation} 
\dot{\textbf{x}}(t) = (c+\frac{1}{T-t}) \left[ 
-\mathcal{R}^{T} \textbf{ F}(\textbf{x}(t)) - \kappa (\mathcal{L}\otimes I_{n}) \textbf{x}(t) \right].  
\label{ClosedLoopSystem}   
\end{equation}

It can be seen that the choice of (\ref{TimeTransformFunction}) yields $d\lambda(s)/ds=T-t$, which appears in (\ref{PrescribedTimeAlgorithm}). By exploiting the transformation function in  (\ref{TimeTransformFunction}), the expression in (\ref{ClosedLoopSystem}) becomes that for $ s \in [0,\infty) $, 
\vspace{-3pt}
\begin{align} 
\psi^{'}(s) &= \lambda^{'}(s) (c+\frac{1}{T-\lambda(s) }) \left[ 
-\mathcal{R}^{T} \textbf{ F}(\psi(s)) - \kappa (\mathcal{L}\otimes I_{n}) \psi(s) \right], \notag \\
&=Te^{-s} (c+\frac{1}{T-\lambda(s) }) \left[ 
-\mathcal{R}^{T} \textbf{ F}(\psi(s)) - \kappa (\mathcal{L}\otimes I_{n}) \psi(s) \right] \notag \\
&=-\theta(s) [\mathcal{R}^{T} \textbf{ F}(\psi(s)) + \kappa (\mathcal{L}\otimes I_{n}) \psi(s)], 
\label{ClosedLoopSysteminSpace}   
\end{align}
where $\theta(s)=(1+cTe^{-s})>0 $ for any $ s \in [0,\infty) $.

The following result shows that the equilibrium of the system  (\ref{ClosedLoopSysteminSpace}) occurs when all players reach a consensus at the NE. 

\vspace{1pt}  
\begin{proposition} \label{proposition1}
Consider the game over $ \mathcal{G} $. Under Assumptions \ref{AssumptionGraph}, \ref{AssumptionConvex}, and  \ref{AssumptionSmooth}, $ \tilde{\textbf{x}}=1_{N}\otimes x^{*} $ is the NE of the noncooperative game if $ \textbf{F}(\tilde{\textbf{x}}) =0_{n}  $ (or $ \triangledown_{i} J_{i}(\tilde{\textbf{x}}^{i}) =0_{n_{i}} $). Further, at the NE,  the estimates of all players reach a consensus and equal to the NE $ x^{*} $, i.e., $ \tilde{\textbf{x}}^{i}=\tilde{\textbf{x}}^{j}= x^{*} $. 
That is, all players' action components coincide with the optimal actions ($ \tilde{\textbf{x}}^{i}_{i}= x^{*}_{i}$). 

\end{proposition}

\begin{proof}
let $ \tilde{\textbf{x}} $ be an equilibrium of the system. 
Then, for $ \tilde{\psi} = \tilde{\textbf{x}}$, if follows from (\ref{ClosedLoopSysteminSpace}) that $ 0_{Nn}=-\mathcal{R}^{T} \textbf{ F}(\tilde{\psi}) - \kappa (\mathcal{L}\otimes I_{n}) \tilde{\psi} $ as $ \theta(s)>0 $, 
which implies that multiplying $ 1^{T}_{N} \otimes I_{n}$ yields  
\vspace{-2pt}
\begin{equation}
0_{n}=- (1^{T}_{N} \otimes I_{n}) \mathcal{R}^{T} \textbf{ F}(\tilde{\psi}) - \kappa (1^{T}_{N} \otimes I_{n}) (\mathcal{L}\otimes I_{n})\tilde{\psi}.   \label{NE4}
\end{equation}  
	
Since $ 1^{T}_{N} \mathcal{L} =0^{T}_{Nn}$ under Assumption \ref{AssumptionGraph}, we obtain $ 0_{n}=(1^{T}_{N} \otimes I_{n}) \mathcal{R}^{T}  \textbf{F}(\tilde{\psi})  $. Then, it follows from the notations of $ \mathcal{R} $ and $ \textbf{F} $ in (\ref{StackedVariable}) that $ \textbf{F}(\tilde{\psi}) =0_{n}  $. Then, submitting it into (\ref{NE4}) gives rise to $ (\mathcal{L}\otimes I_{n}) \tilde{\psi}=0_{Nn} $. Hence, there exists certain $ \theta \in \mathbb{R}^{n} $ such that $ \tilde{\psi}=1_{N}\otimes \theta  $ under Assumption \ref{AssumptionGraph}. Then, it has $ \textbf{F}(1_{N}\otimes \theta) =0_{n}  $ for each player $ i \in \mathcal{V} $. Thus, $ \triangledown_{i} J_{i}(\theta_{i},\theta_{-i}) =0_{n_{i}}$. That is, $ \theta  $ is a unique NE of the game and $ \theta=x^{*} $. Thus, $ \tilde{\textbf{x}}=1_{N}\otimes x^{*} $ and for $i, j \in \mathcal{V}$,  we have $ \tilde{\textbf{x}}^{i}=\tilde{\textbf{x}}^{j}= x^{*} $ (NE of the game).
\end{proof}	
 
Next, we present the main result on the distributed NE seeking of the noncooperative game in a prescribed time.   

\vspace{1pt} 
\begin{theorem} \label{Theorem1}
Suppose that Assumptions \ref{AssumptionGraph}, \ref{AssumptionConvex} and \ref{AssumptionSmooth} hold. Then, given a graph $ \mathcal{G} $ and any initial $ \textbf{x}_{i}(0)$, the proposed prescribed-time distributed NE seeking algorithm in (\ref{PrescribedTimeAlgorithm})   guarantees that the estimates of all players converges to the NE in the user-defined time $ T $, i.e., $ \lim_{t\rightarrow T} \textbf{x}(t) =  \tilde{\textbf{x}}=1_{N}\otimes x^{*}$, provided that   
\vspace{-4pt}
\begin{equation}
\kappa > (\iota^{2}/\varepsilon +\iota ) / \lambda_{2} (\mathcal{L}), \ \forall \iota, \varepsilon>0. \label{GainCondition}
\end{equation}
\end{theorem} 
 
\vspace{-2pt} 
\begin{proof} 
we first make a coordinate transformation as 
\vspace{-3pt}
\begin{align}
\overrightarrow{\psi}(s)  &= (1_{N} \otimes  \mathcal{S} ) \psi(s) \in \mathbb{R}^{Nn} , \ \mathcal{S} =  \frac{1}{N} (1^{T}_{N} \otimes I_{n}), \label{Xright} \\
\overleftarrow{\psi}(s)  &= (\mathcal{T} \otimes I_{n}) \psi(s) \in \mathbb{R}^{Nn}, \ \mathcal{T}=I_{N}-\frac{1}{N} (1_{N }1^{T}_{N} ). \label{Xleft}
\end{align}

\vspace*{-3pt}
Then, it follows from (\ref{Xright}) that the average estimate of $ \psi^{i}(s) $ is described by $ \bar{\psi}(s)=\frac{1}{N}\sum^{N}_{i=1}\psi^{i}(s)= \frac{1}{N} (1^{T}_{N} \otimes I_{n})\psi(s)=\mathcal{S}\psi(s) $. For any $ \psi(s)  \in \mathbb{R}^{Nn}$, it can be decomposed as  $\psi(s) = \overrightarrow{\psi}(s) + \overleftarrow{\psi}(s) $ with $ (\overrightarrow{\psi}(s))^{T}\overleftarrow{\psi}(s) =0 $ and $ (\mathcal{L} \otimes I_{n}) \overrightarrow{\psi}(s) = 0_{Nn} $ using the strongly connected condition in Assumption \ref{AssumptionGraph}. 

For stability analysis, we select the Lyapunov function as
\vspace*{-1pt}
\begin{align}
V(\psi(s))&= \frac{1}{2} (\psi(s)-\tilde{\psi})^{T}(\psi(s)-\tilde{\psi}) \notag \\
& = \frac{1}{2} (\overrightarrow{\psi}(s) + \overleftarrow{\psi}(s)-\tilde{\psi})^{T}(\overrightarrow{\psi}(s) + \overleftarrow{\psi}(s)-\tilde{\psi}) \notag \\
&=\frac{1}{2} 
\left[
\begin{array}{l}  
\overrightarrow{\psi}(s) - \tilde{\psi}   \\ 
\ \ \  \overleftarrow{\psi}(s) 
\end{array}
\right]^{T}  \left[
\begin{array}{l}  
\overrightarrow{\psi}(s) - \tilde{\psi}   \\ 
\ \ \  \overleftarrow{\psi}(s) 
\end{array} 
\right],
\label{LyapunovFunction1}
\end{align}%
where 
$ (\overrightarrow{\psi}(s))^{T}\overleftarrow{\psi}(s) =0 $ and $ (\mathcal{L} \otimes I_{n}) \overrightarrow{\psi}(s) = 0_{Nn} $ are used.

Then, differentiating $ V(\psi(s)) $ with respect to $ s $  yields 
\vspace*{-1pt}
\begin{equation}
V^{'}(\psi(s)) = - \theta(s)(\psi(s)-\tilde{\psi})^{T} [\mathcal{R}^{T}  \textbf{F} (\psi(s)) + \kappa (\mathcal{L}\otimes I_{n})\psi(s) ]. \notag 
\end{equation} 

Since $ 0_{Nn}=- \mathcal{R}^{T} \textbf{F}(\tilde{\psi}) - \kappa (\mathcal{L}\otimes I_{n}) \tilde{\psi} $, the above expression can be further rewritten as 
\vspace*{-1pt}
\begin{equation}
V^{'}(\psi)= -\theta(s) (\psi-\tilde{\psi})^{T} [\mathcal{R}^{T}  (\textbf{F} (\psi)-\textbf{F}(\tilde{\psi})) + \kappa (\mathcal{L}\otimes I_{n}) (\psi-\tilde{\psi}) ]. \notag 
\end{equation} 

In light of $ \tilde{\psi}=1_{N}\otimes x^{*} $, $ \psi=\overrightarrow{\psi}+ \overleftarrow{\psi} $, and $ (\mathcal{L} \otimes I_{n}) \overrightarrow{\psi}= 0_{Nn} $ under Assumption \ref{AssumptionGraph}, then the first term in $ V^{'}(\psi) $  becomes
\vspace*{-2pt}
\begin{align}
&-   (\psi-\tilde{\psi})^{T} \mathcal{R}^{T} [  \textbf{F} (\psi)-\textbf{F}(\tilde{\psi})]   \label{L3} \\
&= -(\overleftarrow{\psi})^{T} \mathcal{R}^{T}  [ \textbf{F} (\psi)-\textbf{F}(\overrightarrow{\psi}) ]  - (\overleftarrow{\psi})^{T} \mathcal{R}^{T}[  \textbf{F}(\overrightarrow{\psi})-\textbf{F}(\tilde{\psi})]  \notag \\
&   - ( \overrightarrow{\psi} - \tilde{\psi})^{T} \mathcal{R}^{T} [ \textbf{F} (\psi)-\textbf{F}(\overrightarrow{\psi})]   -  ( \overrightarrow{\psi} - \tilde{\psi})^{T} \mathcal{R}^{T} [  \textbf{F}(\overrightarrow{\psi})-\textbf{F}(\tilde{\psi})]. \notag 
\end{align}

It follows from Assumption \ref{AssumptionSmooth} that according to the $ \iota_{F} $-Lipschitz continuity of $ F $, it yields that $ \|F(\psi)-F(\overrightarrow{\psi})  \| \leq \iota_{F} \| \overleftarrow{\psi} \| $. Further, $ \| \textbf{F} (\psi)-\textbf{F}(\overrightarrow{\psi})  \| \leq \iota_{\textbf{F}} \| \overleftarrow{\psi} \| $ for certain scalar $ \iota_{\textbf{F}}>0 $. In addition, since $ \| \mathcal{R}^{T} \|=1 $, $ \textbf{F}(\overrightarrow{\psi})=F(\bar{\psi})$,  and $  \textbf{F}(\tilde{\psi})=F(x^{*})=0 $, 
\vspace*{-2pt}
\begin{align}
&-  (\overleftarrow{\psi})^{T} \mathcal{R}^{T}[  \textbf{F}(\overrightarrow{\psi})-\textbf{F}(\tilde{\psi})] \label{L4} \\
&  =-  (\overleftarrow{\psi})^{T} \mathcal{R}^{T} (F(\bar{\psi}) - F(x^{*}))  \leq \iota_{F} T \| \overleftarrow{\psi} \| \| \bar{\psi}- x^{*}\|, \notag \\ 
&-  ( \overrightarrow{\psi} - \tilde{\psi})^{T} \mathcal{R}^{T} [ \textbf{F} (\psi)-\textbf{F}(\overrightarrow{\psi})] \label{L5}  \\  
&= -   (\bar{\psi}-x^{*})^{T} (\textbf{F} (\psi)-\textbf{F}(\overrightarrow{\psi}))   \leq \iota_{\textbf{F}} T \|\bar{\psi}-x^{*} \| \| \overleftarrow{\psi}\|, \notag 
\end{align} 
where the fact that $ \mathcal{R} \overrightarrow{\psi}=\bar{\psi} $ and $ \mathcal{R} \tilde{\psi}=x^{*} $ is used, and exploiting the $ \varepsilon $-strong monotonicity of $ F $, we can further  have that  
\vspace*{-2pt}
\begin{align}
& - ( \overrightarrow{\psi} - \tilde{\psi})^{T} \mathcal{R}^{T} [  \textbf{F}(\overrightarrow{\psi})-\textbf{F}(\tilde{\psi})] \notag \\
&  =-  (\bar{\psi}-x^{*})^{T}(F(\bar{\psi})-F(x^{*})) \leq -\varepsilon   \| \bar{\psi}-x^{*} \|^{2}. \label{L6}
\end{align} 

In addition, the second term in (\ref{L2}) can be rewritten as
\vspace*{-3pt}
\begin{align}
&-(\psi-\tilde{\psi})^{T}(\mathcal{L} \otimes I_{n}) (\psi-\tilde{\psi})  = -(\overrightarrow{\psi}+\overleftarrow{\psi} )^{T}(\mathcal{L}\otimes I_{n}) (\overrightarrow{\psi}+\overleftarrow{\psi})     \notag \\
& = - \overleftarrow{\psi}^{T} (\mathcal{L} \otimes I_{n}) \overleftarrow{\psi}  \leq - \lambda_{2} (\mathcal{L}) \| \overleftarrow{\psi} \|^{2}.
\label{L7}
\end{align} 

Let $ \iota=\max\{\iota_{F},\iota_{\textbf{F}}\} $. Substituting (\ref{L4})-(\ref{L7}) into (\ref{L2}) gives 
\vspace*{-4pt}
\begin{align}
\hspace{-0.3em}
V^{'} & \leq \theta [ 2 \iota \| \overleftarrow{\psi} \| \| \bar{\psi} -x^{*} \|  -  (\kappa \lambda_{2} ( \hat{\mathcal{L}}) - \iota ) \| \overleftarrow{\psi} \|^{2}  - \varepsilon \|\bar{\psi} -x^{*} \|^{2} ] \notag \\ 
& = - \theta \left[
\begin{array}{l}  
\|\bar{\psi}-x^{*}\|  \\ 
\ \ \|\overleftarrow{\psi} \|
\end{array}
\right]^{T}      
\left[
\begin{array}{cc} 
\varepsilon   & - \iota \\
-\iota  & \kappa \lambda_{2} (\mathcal{L}) - \iota 
\end{array}
\right]
\left[
\begin{array}{l}  
\| \bar{\psi}-x^{*} \|  \\ 
\ \  \|\overleftarrow{\psi} \|
\end{array}
\right], \notag \\
& = - \frac{\theta}{2} \left[
\begin{array}{l}  
\|\overrightarrow{\psi}-\tilde{\psi} \| \\ 
\ \  \|\overleftarrow{\psi} \|
\end{array}
\right]^{T}     
O
\left[
\begin{array}{l}  
\|\overrightarrow{\psi}-\tilde{\psi} \| \\ 
\ \  \|\overleftarrow{\psi} \|
\end{array}
\right],
\label{L8}
\end{align} 
where the fact that $ \|\bar{\psi}-x^{*}\|=\frac{1}{\sqrt{N}} \|\overrightarrow{\psi}-\tilde{\psi} \| $ is used, $ \theta>0 $, and $O = 2 \left[
\begin{array}{cc} 
\frac{\varepsilon }{N} & - \frac{\iota}{\sqrt{N}} \\
-\frac{\iota}{\sqrt{N}} & \kappa  \lambda_{2} (\mathcal{L}) - \iota 
\end{array}
\right] >0 $ if $ \kappa > \frac{1}{\lambda_{2} (\mathcal{L}) } (\frac{\iota^{2}}{\varepsilon } +\iota) $. 	

Since $V( \psi(s) ) $ in (\ref{LyapunovFunction1}) is bounded on the interval $ s\in [0,\infty) $ and $ V^{'}( \psi(s) ) < 0 $ by (\ref{L8}), the closed-loop system in (\ref{ClosedLoopSysteminSpace}) is globally stable for any $ \psi(0) $ when $ s \rightarrow \infty $. Hence, the estimate states $  \overrightarrow{\psi}(s)-\tilde{\psi} $ and $\overleftarrow{\psi}(s)$ are bounded, and converge a largest invariant set $ \mathcal{M}= \{\overrightarrow{\psi}(s)=\tilde{\psi} \ \text{and} \ \overleftarrow{\psi}(s) = 0_{Nn}\} $ based on the LaSalle’s invariance principle \cite{BookKhail}. Thus, on this invariant set,   $ \lim_{t\rightarrow \infty} (\overrightarrow{\psi}(s)-\tilde{\psi})=0_{Nn}$ and $ \lim_{t\rightarrow \infty} \overleftarrow{\psi}(s) = 0_{Nn}$. According to the coordinate transformation $ \overrightarrow{\psi}(s) $ and $ \overleftarrow{\psi}(s) $ in (\ref{Xright})-(\ref{Xleft}), and the fact that $ \psi(s)= \overrightarrow{\psi}(s)+ \overleftarrow{\psi(s)}$, we obtain $ \lim_{t\rightarrow \infty} (\psi(s)-\tilde{\psi})=0_{Nn}$. 
Finally, based on the time transformation function in (\ref{TimeTransformFunction}) with the fact that $ t\rightarrow T $ as $ s\rightarrow\infty $, and $ \psi(s) \triangleq \textbf{x}(t) $ with $ \textbf{x}(t) $ being the solution to (\ref{ClosedLoopSystem}), we conclude that $ \lim_{t\rightarrow T} \tilde{\textbf{x}} = 1 \otimes x^{*} $. Thus, Problem \ref{Problem} is solved in a prescribed-time $T$. 
\end{proof}

\vspace{2pt} 
Next, we show that the NE can be maintained and the control input signal remains zero over $[T, \infty)$. Moreover, the this control signal remains $ \mathcal{C}^{1} $ smooth and uniformly bounded over the whole time interval $[0, \infty)$. The following theorem summaries this result.  
 
\vspace{2pt} 
\begin{theorem} \label{The1Boundness}
Consider the game over $ \mathcal{G} $. Under Assumptions \ref{AssumptionGraph}, \ref{AssumptionConvex} and \ref{AssumptionSmooth}, the proposed prescribed-time distributed NE seeking algorithm guarantees that the NE is  maintained and the control signal remains zero for $ t \geq T $. Further, the control signal remains $ \mathcal{C}^{1} $ smooth and uniformly bounded for any $t\in [0, \infty)$.  	
\end{theorem} 

\begin{proof}
It follows from (\ref{ClosedLoopSysteminSpace}) that for a stacked vector $ \vartheta(s)$,  
\vspace{-3pt}
\begin{align}
\frac{d \psi(s)}{ds}&=\vartheta(s)=-\theta(s) [\mathcal{R}^{T} \textbf{ F}(\psi(s)) + \kappa (\mathcal{L}\otimes I_{n}) \psi(s)],  \label{L9} \\
\frac{d\vartheta(s)}{ds}&=-cTe^{-s} [-\mathcal{R}^{T} \textbf{F}(\psi(s)) - \kappa (\mathcal{L}\otimes I_{n}) \psi(s)] \notag \\
& \ \ \ -\theta(s)  [ \mathcal{R}^{T} \frac{ \partial{\textbf{F}(\psi(s))}} {\partial{\psi(s)}}+ \kappa (\mathcal{L}\otimes I_{n}) ]\frac{d \psi(s)}{ds} \notag \\
&= \ \ \ -(1-\frac{1}{\theta(s)})\vartheta(s) - \theta(s) H(\psi(s)) \vartheta(s), \label{L10} 
\end{align} 
where $ H(\psi(s)) = \mathcal{R}^{T} \frac{ \partial{\textbf{F}(\psi(s))}} {\partial{\psi(s)}}+ \kappa (\mathcal{L}\otimes I_{n})   $ denotes a matrix with $  \partial{\textbf{F}(\psi(s))} / \partial{\psi(s)} $ being a  Hessian matrix, and $  \frac{1}{\theta(s)} \in (0,1) $.

From Theorem \ref{Theorem1}, it has been proved by (\ref{L8}) that 
\vspace{-3pt}
\begin{equation}
V^{'}(s) \leq - \lambda_{min}(O)\theta(s) V(s),  
\label{L11}
\end{equation}  
which implies that the error satisfies 
\vspace{-3pt}
\begin{equation}
\|\psi(s)-\tilde{\psi}\| \leq e^{-\lambda_{min}(O) (s-s_{0})} \|\psi(s_{0})-\tilde{\psi}\|, \ \forall s\geq s_{0}=0. \label{L12}
\end{equation}

Hence, the boundness of $ \vartheta(s) $ can be derived as 
\begin{align}
\hspace{-0.4em}
\|\vartheta(s)\| &=\theta(s) \|\mathcal{R}^{T}(\textbf{F}(\psi(s))-F(\tilde{\psi})) + \kappa (\mathcal{L}\otimes I_{n}) (\psi(s)-\tilde{\psi}))\|,  \notag \\
& \leq \theta(s) [\iota_{\textbf{F}}\| \|\psi(s)-\tilde{\psi}\| +k\|\mathcal{L}\|\psi(s)-\tilde{\psi}\|] \notag \\
& \leq \theta(s) (\iota_{\textbf{F}}+k\|\mathcal{L}\|)  e^{-\lambda_{min}(O)} \|\psi(s(0))-\tilde{\psi}\|, 
\label{L13} 
\end{align}
which implies that $ \vartheta(s) $ is bounded on $[0,\infty)$, i.e., $ \vartheta(s) \in \mathcal{L}_{\infty} $ on $[0,\infty)$. Further, it follows from (\ref{L10}) that $ \vartheta(s) $ is a bounded solution to the system in (\ref{L10}) on $[0,\infty)$. Thus, it is concluded that $ \textbf{u}^{i}(t) $ in (\ref{AlgorithmStatic}) is bounded  on $[0,T)$.  Moreover, according to (\ref{L9}) and (\ref{L10}), we can see that both $ \vartheta(s) $ and $ d\vartheta(s)/ds $ are continuous with respect to $ \psi(s) $ on $ [0,\infty) $. Since $ \psi(s) $ is continuous with respect to $ s $  based on the continuity of (\ref{L9}), we obtain that $ \vartheta(s) $ and $ d\vartheta(s)/ds $ are continuous with respect to $ s\in [0,\infty) $. Thus, $ \psi(s) $ is $ \mathcal{C}^{1} $ smooth with respect to $ s\in [0,\infty) $. That is, $ \textbf{u}^{i}(t) $ in (\ref{AlgorithmStatic}) is $ \mathcal{C}^{1} $ smooth with respect to $t\in [0,T) $. 

\vspace{2pt}
Next, we show that the NE can be maintained and the control input signal $ \textbf{u}^{i}(t) $ remains zero over $ [0,\infty) $. Notice that for $ t\in [T,\infty) $, the proposed algorithm in (\ref{AlgorithmStatic}) yields 
\vspace{-2pt}
\begin{equation} 
\dot{\textbf{x}}(t) = c \left[ 
-\mathcal{R}^{T} \textbf{ F}(\textbf{x}(t)) - \kappa (\mathcal{L}\otimes I_{n}) \textbf{x}(t) \right],   
\ t\in [T,\infty). 
\label{ClosedLoopSystem11}   
\end{equation}

Choose a similar Lyapunov function as  (\ref{LyapunovFunction1}) for $ t\in [T,\infty) $.   
\vspace{-4pt}
\begin{align}
V(\textbf{x})&= \frac{1}{2} (\textbf{x}-\tilde{\textbf{x}})^{T}(\textbf{x}-\tilde{\textbf{x}})= \frac{1}{2} (\overrightarrow{\textbf{x}}+\overleftarrow{\textbf{x}} -\tilde{\textbf{x}})^{T}(\overrightarrow{\textbf{x}}+\overleftarrow{\textbf{x}} -\tilde{\textbf{x}}) \notag \\
&=\frac{1}{2} 
\left[
\begin{array}{l}  
\overrightarrow{\textbf{x}}-\tilde{\textbf{x}}    \\ 
\ \ \   \overleftarrow{\textbf{x}}  
\end{array}
\right]^{T} \left[
\begin{array}{l}  
\overrightarrow{\textbf{x}}-\tilde{\textbf{x}}   \\ 
\ \  \  \overleftarrow{\textbf{x}}  
\end{array}
\right], \ t\in [T,\infty) 
\label{LyapunovFunction111}
\end{align}%
where $ \overrightarrow{\textbf{x}} = (1_{N} \otimes  \mathcal{S} )\textbf{x} \in \mathbb{R}^{Nn} $ and $ \overleftarrow{\textbf{x}} = (\mathcal{T} \otimes I_{n}) \textbf{x} \in \mathbb{R}^{Nn} $ with $ \mathcal{S}, \mathcal{T} $ being defined in (\ref{Xright})-(\ref{Xleft}).  
 
Then, the time derivative of $ V(\textbf{x}) $ along (\ref{LyapunovFunction111}) yields  
\vspace*{-4pt}
\begin{align}
\hspace{-0.3em}
\dot{V}(\textbf{x}) & \leq   - c  \left[
\begin{array}{l}  
\|\bar{\textbf{x}}-x^{*}\|  \\ 
\ \ \|\overleftarrow{\textbf{x}} \|
\end{array}
\right]^{T}      
\left[
\begin{array}{cc} 
\varepsilon   & - \iota \\
-\iota  & \kappa \lambda_{2} (\mathcal{L}) - \iota 
\end{array}
\right]
\left[
\begin{array}{l}  
\| \bar{\textbf{x}}-x^{*} \|  \\ 
\ \  \|\overleftarrow{\textbf{x}} \|
\end{array}
\right], \notag \\
& = - \frac{c}{2} \left[
\begin{array}{l}  
\|\overrightarrow{\textbf{x}}-\tilde{\textbf{x}} \| \\ 
\ \  \|\overleftarrow{\textbf{x}} \|
\end{array}
\right]^{T}     
O
\left[
\begin{array}{l}  
\|\overrightarrow{\textbf{x}}-\tilde{\textbf{x}} \| \\ 
\ \  \|\overleftarrow{\textbf{x}} \|
\end{array}
\right], \ t\in [T,\infty) 
\label{LLV}
\end{align} 
which implies that
\vspace*{-3pt} 
\begin{equation}
\dot{V}(\textbf{x}) \leq - c \lambda_{min}(O)  V(\textbf{x}) \leq 0,  \ t\in [T,\infty).   
\label{LLV1}
\end{equation}

Since $ \textbf{x} $ is continuous at $ t=T $ from the continuity of system, we obtain that $  V(\textbf{x}) $ is continuous at $ t=T $, and then 
\vspace*{-3pt}
\begin{equation}
V(T)= \lim_{t\rightarrow T^{-}} (\textbf{x}-\tilde{\textbf{x}})^{T}(\textbf{x}-\tilde{\textbf{x}}) =0.    
\label{LLV2}
\end{equation}

Combining (\ref{LLV1}) and (\ref{LLV2}) gives rise to 
\vspace*{-3pt}
\begin{equation}
0\leq V(t) \leq V(T) =0,   \ t\in [T,\infty), 
\label{LLV3}
\end{equation}
which implies that $ V(t) \equiv 0 $ on $ [T,\infty) $. Thus, $ \textbf{x}-\tilde{\textbf{x}} \equiv 0$ on $ [T,\infty) $. Based on the fact that  $\dot{\textbf{x}}^{i}(t)=c\textbf{u}^{i}(t) $ for $ t \geq T$, it can be verified that $\dot{\textbf{x}}^{i}(t)= c \textbf{u}^{i}(t) \equiv 0$ on $ [T,\infty) $. Hence, the NE is maintained and the control input signal remains zero over $[T,\infty)$. Further, it can be shown that the control input signal is $ \mathcal{C}^{1} $ smooth and uniformly bounded on $[0,\infty)$.

Overall, it is concluded that the NE is found in a prescribed-time $ T $ and is maintained over $[T,\infty)$, and moreover, the control input signal is $ \mathcal{C}^{1} $ smooth and uniformly bounded on $[0,\infty)$.	
\end{proof}

\subsection{Prescribed-Time And Fully Distributed NE Seeking Design} 
In this subsection, we develop a novel NE seeking algorithm that is fully distributed to remove the strong requirement that $ \kappa $ in (\ref{GainCondition}) requires the known global graph information.

\vspace{3pt}
\textit{\textbf{Prescribed-time and fully distributed NE seeking algorithm}}: instead of using the static gain $ \kappa $ in (\ref{PrescribedTimeAlgorithm}), we propose a dynamic gain based distributed integral loop that tunes on-line the weights on the edges of the topology. In light of (\ref{NECompactDynamics}), the prescribed-time and fully distributed NE seeking algorithm is described by    
\vspace{-4pt}
\begin{align}
\dot{\textbf{x}}^{i}(t)&= (c+\frac{1}{T-t}) \textbf{u}^{i}(t), \ \textbf{u}^{i}(t)= -\mathcal{R}^{T}_{i} \triangledown_{i} J_{i}(\textbf{x}^{i}(t))  + \textbf{e}^{i}(t),   \notag \\
\textbf{e}^{i}(t)&=- \sum^{N}_{j=1} \kappa_{ij}(t) a_{ij}  (\textbf{x}^{i}(t)-\textbf{x}^{j}(t)),  
 i  \in \mathcal{V},   t\in [0,T), 
\label{PrescribedTimeAlgorithmAdaptive} \\
\dot{\kappa}_{ij}(t)& =(c+\frac{1}{T-t})  \left[ \gamma_{ij} a_{ij} (\textbf{x}^{i}(t)-\textbf{x}^{j}(t))^{T} (\textbf{x}^{i}(t)-\textbf{x}^{j}(t) )\right] , \notag 
\end{align}  
where $\kappa_{ij}(t)$ is a dynamic control gain with $ \kappa_{ij}(0)=\kappa_{ji}(0) \geq 0 $, $  i, j   \in \mathcal{V} $, $\gamma_{ij}=\gamma_{ji} >0$ is a scalar, and for $ t \geq T$, $\dot{\textbf{x}}^{i}(t)=c\textbf{u}^{i}(t)$, $\dot{\kappa}_{ij}(t)=c\gamma_{ij} a_{ij} (\textbf{x}^{i}(t)-\textbf{x}^{j}(t))^{T} (\textbf{x}^{i}(t)-\textbf{x}^{j}(t) ) $.  

\vspace{2pt} 
\begin{remark}
As can be seen that in (\ref{PrescribedTimeAlgorithmAdaptive}), only relative estimated information is used, and this proposed NE seeking algorithm will be proved to not require any global information on the algebraic connectivity of graphs, the Lipschitz and monotone constants of pseudo-gradients and the number of players. Unlike [8]-[16], the developed NE seeking algorithm is thus fully distributed.  
\end{remark}

\vspace{3pt}
Then, based on $ t=\lambda(s) \triangleq T(1-e^{-s}) $ in (\ref{TimeTransformFunction}), we can transform (\ref{PrescribedTimeAlgorithmAdaptive}) into the  closed-loop system with respect to $ s \in [0,\infty) $,  
\vspace{-4pt} 
\begin{align}
\frac{d\psi^{i}(s)}{ds}&=   - \theta(s)[\mathcal{R}^{T}_{i} \triangledown_{i} J_{i}(\psi^{i}(s)) + \sum^{N}_{j=1} \phi_{ij}(s) a_{ij}  (\psi^{i}(s)-\psi^{j}(s))],  \notag \\
\frac{d\phi_{ij}(s)}{ds}& =\theta(s) \gamma_{ij} a_{ij}  (\psi^{i}(s)-\psi^{j}(s))^{T} (\psi^{i}(s)-\psi^{j}(s)), 
\label{PrescribedTimeAlgorithmAdaptiveClosed}
\end{align}  
where  $ \theta(s)>0 $, and $ \psi^{i}(s) $, $ \phi^{i}(s) $ are the transformation functions with respect to $ \textbf{x}^{i}(t) $, $ \kappa_{ij}(t) $ in (\ref{PrescribedTimeAlgorithmAdaptive}), respectively. 

\vspace{2pt}
Next, we present the result on the fully distributed NE seeking of the noncooperative game in a prescribed time.   

\vspace{2pt} 
\begin{theorem} \label{Theorem2}
Suppose that Assumptions \ref{AssumptionGraph}, \ref{AssumptionConvex} and \ref{AssumptionSmooth} hold. Then, given a graph $ \mathcal{G} $ and any initial $ \textbf{x}_{i}(0)$, the proposed prescribed-time and fully distributed algorithm in (\ref{PrescribedTimeAlgorithmAdaptive}) ensures that 
	
\begin{enumerate}
\item the estimates of all players converges to the NE in the user-defined time $ T $, i.e., $ \lim_{t\rightarrow T} \textbf{x}(t) =  \tilde{\textbf{x}}=1_{N}\otimes x^{*}$; 
\item the dynamic gain $ \kappa_{ij}(t) $ for $\forall \ i,j \in \mathcal{V} $ is monotonically increasing and converges to certain finite constants. 
\end{enumerate}	 
\end{theorem}

\begin{proof}
we choose the following Lyapunov function candidate 
\vspace{-3pt}
\begin{equation}
W(s)=V(\psi(s)) + U(s), \label{Lyapunovfunction2}
\end{equation}
where both differentiable function  $ V(\psi(s)) $, $ U(s) $ are given by 
\vspace{-3pt}
\begin{equation}
V(\psi(s)) = \frac{1}{2} \sum_{i=1}^{N} (\psi^{i}(s)-x^{*})^{T}(\psi^{i}(s)-x^{*}),
\end{equation}
\vspace{-3pt}
\begin{equation}
U(s) =  \sum_{i=1}^{N} \sum_{j \in \mathcal{N}_{i}} \frac{1}{4\gamma_{ij}}  (\phi_{ij}(s)-\gamma^{*})^{T}(\phi_{ij}(s)-\gamma^{*}), 
\end{equation}
where $ \gamma^{*}>0 $ is a constant parameter to be determined later. 

Next, differentiating $ W(s) $ with respect to $ s $ and exploiting (\ref{PrescribedTimeAlgorithmAdaptiveClosed}) give rise to the following expression  
\vspace*{-2pt}
\begin{align}
\hspace{-0.5em}
W^{'}(s) &= -\theta(s)  \sum_{i=1}^{N} (\psi^{i}(s)-x^{*})^{T}   \mathcal{R}^{T}_{i} [\triangledown_{i} J_{i}(\psi^{i}(s)) -\triangledown_{i} J_{i}(x^{*}) ] \notag \\
& -\theta(s) \sum_{i=1}^{N}\sum^{N}_{j=1} \phi_{ij}(s) a_{ij}  (\psi^{i}(s)-x^{*})^{T}   (\psi^{i}(s)-\psi^{j}(s)) \notag \\
&+\frac{\theta(s)}{2}\sum_{i=1}^{N} \sum_{j \in \mathcal{N}_{i}} a_{ij}  (\phi_{ij}(s)-\gamma^{*})  ||\psi^{i}(s)-\psi^{j}(s)||^{2},
\label{NL1}
\end{align} 
where the fact that $ 0_{n}=\mathcal{R}^{T}_{i} \triangledown_{i} J_{i}(x^{*}) $ at the NE has been used.

Notice that the last term in (\ref{NL1}) can be expressed as 
\vspace*{-2pt}
\begin{align}
& \ \ \ \frac{1}{2}\sum_{i=1}^{N} \sum_{j \in \mathcal{N}_{i}} a_{ij}  (\phi_{ij}(s)-\gamma^{*})  ||\psi^{i}(s)-\psi^{j}(s)||^{2} \notag \\
&= \frac{1}{2}\sum_{i=1}^{N} \sum_{j \in \mathcal{N}_{i}} a_{ij}  (\phi_{ij}(s)-\gamma^{*}) (\psi^{i}(s)-\psi^{j}(s))^{T} (\psi^{i}(s)-\psi^{j}(s)) \notag \\
&= \sum_{i=1}^{N} \sum_{j \in \mathcal{N}_{i}} a_{ij} (\phi_{ij}(s)-\gamma^{*}) (\psi^{i}(s))^{T}  (\psi^{i}(s)-\psi^{j}(s)),
\label{NL2}
\end{align}
where the fact that $ \phi_{ij}(s)=\phi_{ji}(s) $ for $ \forall s \geq 0 $ under Assumption \ref{AssumptionGraph} has been used to obtain the last term.   

\vspace*{3pt}
Since $ \sum_{i=1}^{N}\sum^{N}_{j=1}a_{ij}  \phi_{ij}(s)  (x^{*})^{T} (\psi^{i}(s)-\psi^{j}(s)) =0$, then combining (\ref{NL1})-(\ref{NL2}) and canceling the same term give rise to the following expression in the sense of compact from 
\vspace*{-3pt}
\begin{align}
W^{'}(s) =& \theta(s) \sum_{i=1}^{N} (\psi^{i}(s)-x^{*})^{T}   \mathcal{R}^{T}_{i} [\triangledown_{i} J_{i}(\psi^{i}(s)) -\triangledown_{i} J_{i}(x^{*}) ] \notag \\
& -\theta(s) \gamma^{*} \sum_{i=1}^{N} \sum_{j \in \mathcal{N}_{i}} a_{ij}  (\psi^{i}(s))^{T}  (\psi^{i}(s)-\psi^{j}(s)) \notag \\
=& -\theta(s)  (\psi(s)-\tilde{\psi})^{T} \mathcal{R}^{T}   [\textbf{F} (\psi(s)) - \textbf{F} (\tilde{\psi})] \notag \\
& -\theta(s)  \gamma^{*} \psi^{T}(s) (\mathcal{L} \otimes I_{n}) \psi(s).
\label{NL3}
\end{align}

Next, drop the symbol $ s $ in the following analysis for simplicity, and then, it follows from (\ref{L3})-(\ref{L5}) that  
\vspace*{-1pt} 
\begin{align}
- (\psi -\tilde{\psi})^{T} \mathcal{R}^{T}   [\textbf{F} (\psi ) - \textbf{F} (\tilde{\psi})] &  \leq \iota  \| \overleftarrow{\psi} \|^{2}  - \varepsilon  \|\bar{\psi} -x^{*} \|^{2} \notag \\ 
& \ \ + 2 \iota \| \overleftarrow{\psi } \| \| \bar{\psi} -x^{*} \|. \label{NL4}
\end{align}

Further, the term  $- \gamma^{*} \psi^{T} (\mathcal{L} \otimes I_{n}) \psi $ can be expressed as 
\vspace*{-1pt}
\begin{align}
&-\gamma^{*} \psi^{T} (\mathcal{L} \otimes I_{n}) \psi =- \gamma^{*} (\psi-\tilde{\psi})^{T}(\mathcal{L} \otimes I_{n}) (\psi-\tilde{\psi})  \notag \\
& = - \gamma^{*} (\overrightarrow{\psi}+\overleftarrow{\psi} )^{T}(\mathcal{L}\otimes I_{n})  (\overrightarrow{\psi}+\overleftarrow{\psi})     
\leq -\gamma^{*}\lambda_{2} (\mathcal{L}) \| \overleftarrow{\psi} \|^{2}.
\label{NL5}
\end{align} 

Thus, combining (\ref{NL3})-(\ref{NL5}) gives rise to 
\vspace*{-1pt}
\begin{align}
\hspace{-0.5em}
W^{'} & \leq \theta [2 \iota \| \overleftarrow{\psi } \| \| \bar{\psi} -x^{*} \|-  (\gamma^{*} \lambda_{2} (\mathcal{L}) - \iota ) \| \overleftarrow{\psi} \|^{2}  - \varepsilon  \|\bar{\psi} -x^{*} \|^{2} ] \notag \\
& = - \theta \left[
\begin{array}{l}  
\|\bar{\psi}-x^{*}\|  \\ 
\ \ \|\overleftarrow{\psi} \|
\end{array}
\right]^{T}      
\left[
\begin{array}{cc} 
\varepsilon   & - \iota \\
-\iota  & \gamma^{*} \lambda_{2} (\mathcal{L}) - \iota 
\end{array}
\right]
\left[
\begin{array}{l}  
\| \bar{\psi}-x^{*} \|  \\ 
\ \  \|\overleftarrow{\psi} \|
\end{array}
\right], \notag \\
& = -\theta \left[
\begin{array}{l}  
\|\overrightarrow{\psi}-\tilde{\psi} \| \\ 
\ \  \|\overleftarrow{\psi} \|
\end{array}
\right]^{T}     
P
\left[
\begin{array}{l}  
\|\overrightarrow{\psi}-\tilde{\psi} \| \\ 
\ \  \|\overleftarrow{\psi} \|
\end{array}
\right], 
\label{NL6}
\end{align} 
where $P=\left[
\begin{array}{cc} 
\frac{\varepsilon } {N} & - \frac{\iota}{\sqrt{N}} \\
-\frac{\iota}{\sqrt{N}} & \gamma^{*} \lambda_{2} (\mathcal{L}) - \iota
\end{array}
\right] >0 $  if $ \gamma^{*}  > \frac{1}{\lambda_{2} (\mathcal{L}) } (\frac{\iota^{2}}{\varepsilon } +\iota)  $.

Since $W( \psi(s)) $ in (\ref{Lyapunovfunction2}) is bounded on the interval $ s\in [0,\infty) $ and $ W^{'}( \psi(s) ) < 0 $ by (\ref{NL6}), the closed-loop system in (\ref{PrescribedTimeAlgorithmAdaptiveClosed}) is globally stable for any $ \psi(0) $ on $ s\in [0,\infty) $. Hence, the estimate states $  \overrightarrow{\psi}(s)-\tilde{\psi} $ and $\overleftarrow{\psi}(s)$ are bounded, and converge the largest invariant set $ \mathcal{M}= \{\overrightarrow{\psi}(s)=\tilde{\psi} \ \text{and} \ \overleftarrow{\psi}(s) = 0_{Nn}\} $.
Moreover, the dynamic gain $ \phi_{ij}(s) $ is bounded and converges to certain finite values. The rest is similar to the analysis in the proof of Theorem 1 and it is concluded to $ \lim_{t\rightarrow T} \tilde{\textbf{x}} = 1 \otimes x^{*} $. 
\end{proof}

Similarly, we can show that the proposed prescribed-time and fully distributed NE seeking algorithm ensures that the NE can be maintained and the control input signal remains zero over $[T, \infty)$. Moreover, this signal remains $ \mathcal{C}^{1} $ smooth and uniformly bounded on the whole time interval  $[0, \infty)$. The details are similar to the proof of Theorem \ref{The1Boundness}, and thus are omitted here. 

%

\vspace{-8pt}
\subsection{Prescribed-Time And Fully Distributed NE Seeking Design on Jointly Switching Communication Topologies}	
In this subsection, we further extend the NE seeking algorithm in (\ref{PrescribedTimeAlgorithmAdaptive}) to consider jointly switching topologies.


\vspace{3pt}
\textit{\textbf{Prescribed-time and fully distributed NE seeking design over jointly switching graphs}}: in light of (\ref{PrescribedTimeAlgorithmAdaptive}), we further propose the following novel algorithm described by 
\vspace{-3pt} 
\begin{align}
\hspace{-0.5em}
\dot{\textbf{x}}^{i}(t)&= (c+\frac{1}{T-t}) \textbf{u}^{i}(t), \textbf{u}^{i}(t)= -\mathcal{R}^{T}_{i} \triangledown_{i} J_{i}(\textbf{x}^{i}(t))  + \textbf{e}^{i}(t),  \notag \\
\hspace{-0.5em}
\textbf{e}^{i}(t)&=- \sum^{N}_{j=1} \kappa_{ij}(t) a^{\sigma(t)}_{ij}  (\textbf{x}^{i}(t)-\textbf{x}^{j}(t)),   i \in \mathcal{V},  t\in [0,T), 
\label{PrescribedTimeAlgorithmAdaptiveSwitching} \\
\hspace{-0.5em}
\dot{\kappa}_{ij}(t)& =(c+\frac{1}{T-t}) \left[ \gamma_{ij} a^{\sigma(t)}_{ij}  (\textbf{x}^{i}(t)-\textbf{x}^{j}(t))^{T} (\textbf{x}^{i}(t)-\textbf{x}^{j}(t) )\right] , \notag 
\end{align} 
where $ a^{\sigma(t)}_{ij} $ represents the adjacency element of jointly switching topologies $ \mathcal{G}^{\sigma(t)} $, and for $ t\geq T $, $ \dot{\textbf{x}}^{i}(t)=c\textbf{u}^{i}(t), \dot{\kappa}_{ij}(t)=\gamma_{ij} a^{\sigma(t)}_{ij} \\  (\textbf{x}^{i}(t)-\textbf{x}^{j}(t))^{T} (\textbf{x}^{i}(t)-\textbf{x}^{j}(t) )$.

Next,
we transform (\ref{PrescribedTimeAlgorithmAdaptiveSwitching}) into the following closed-loop error system with respect to the variable $ s \in [0,\infty) $,  
\vspace{-3pt}
\begin{align}
\hspace{-0.8em}
\frac{d\psi^{i}(s)}{ds}&=   \theta(s) [\sum^{N}_{j=1} \phi_{ij}(s) a^{\alpha(s)}_{ij}  (\psi^{j}(s)-\psi^{i}(s)) -  \mathcal{R}^{T}_{i} \triangledown_{i} J_{i}(\psi^{i}(s))] ,  \notag \\
\hspace{-0.8em}
\frac{d\phi_{ij}(s)}{ds}& = \theta(s) \gamma_{ij} a^{\alpha(s)}_{ij} (\psi^{i}(s)-\psi^{j}(s))^{T} (\psi^{i}(s)-\psi^{j}(s)), 
\label{PrescribedTimeAlgorithmAdaptiveClosedSwitching}
\end{align}  
where $  \theta(s)>0 $, and $ \psi^{i}(s) $, $ \phi_{ij}(s) $, $ a^{\alpha(s)}_{ij}$ are transformation functions with respect to $ \textbf{x}^{i}(t) $, $ \kappa_{ij}(t) $, and $ a^{\sigma(t)}_{ij} $ in (\ref{PrescribedTimeAlgorithmAdaptiveSwitching}), respectively.


Next, we present the result on the prescribed-time and fully distributed NE seeking over switching graphs.   

\vspace{2pt} 
\begin{theorem} \label{Theorem3}
Suppose that Assumptions \ref{AssumptionGraphSwitching}-\ref{AssumptionSmooth} hold. Given the graph $ \mathcal{G}^{\sigma(t)} $ and any initial $ \textbf{x}_{i}(0)$, the proposed prescribed-time fully distributed algorithm in (\ref{PrescribedTimeAlgorithmAdaptiveSwitching}) ensures that 
not only all players' estimates converge to the NE in a time $ T $, i.e., $ \lim_{t\rightarrow T} \textbf{x}(t) =  \\ \tilde{\textbf{x}}=1_{N}\otimes x^{*}$ but also the dynamic gain $ \kappa_{ij}(t) $  is monotonically increasing and converges to certain finite constants.  
\end{theorem}  	 
	 
\begin{proof}
consider a same Lyapunov function candidate  in (\ref{Lyapunovfunction2}). 
Then, differentiating $ W(s) $ with respect to $ s $ yields  
\vspace*{-3pt}
\begin{align}
W^{'}(s) & \leq    - \theta(s) \gamma^{*}  (\psi(s)- \tilde{\psi} )^{T} (\mathcal{L}^{\alpha(s)} \otimes I_{n})  (\overrightarrow{\psi}(s)+\overleftarrow{\psi}(s))  \notag \\
&  + \theta(s) [2 \iota \| \overleftarrow{\psi }(s) \| \| \bar{\psi} -x^{*} \|+ \iota  \| \overleftarrow{\psi}(s) \|^{2} - \varepsilon  \|\bar{\psi}(s) -x^{*} \|^{2} ],   \notag
\end{align}
where $ \mathcal{L}^{\alpha(s)} $ is the time transformation Laplacian matrix of $ \mathcal{L}^{\sigma(t)} $. 

Based on Assumption \ref{AssumptionGraphSwitching}, $ \mathcal{L}^{\sigma(t)} $ is jointly connected and time-invariant on each time interval $ [t_{i}, t_{i+1}) $, $ i=0,1,2,\cdots $. That is, there exists certain new time interval $ [s_{i}, s_{i+1}) $, $  i=0,1,2,\cdots $ so that $ \mathcal{L}^{\alpha(s)}=\mathcal{L}^{\sigma(t)} $ as $ \mathcal{L}^{\sigma(t)} $ is time-invariant on each time interval  $ [s_{i}, s_{i+1}) $. Moreover, $ \mathcal{L}^{\alpha(s)}$ is jointly connected on each interval $ [s_{i}, s_{i+1}) $.  Thus, we have that for each interval  $ s \in [s_{i}, s_{i+1}) $, 
\vspace*{-1pt}
\begin{align}
\hspace{-0.5em}
W^{'}(s) & \leq  - \theta(s) \gamma^{*} 
(\overrightarrow{\psi}(s)+\overleftarrow{\psi}(s) )^{T} (\mathcal{L}^{\alpha(s)} \otimes I_{n})  (\overrightarrow{\psi}(s)+\overleftarrow{\psi}(s))  
\notag \\
& \ \ + \theta(s) [2 \iota \| \overleftarrow{\psi }(s) \| \| \bar{\psi} -x^{*} \|+ \iota  \| \overleftarrow{\psi}(s) \|^{2} - \varepsilon  \|\bar{\psi}(s) -x^{*} \|^{2} ]  \notag \\
& = -  \theta(s)\gamma^{*}   (\overleftarrow{\psi}(s) )^{T} (\mathcal{L}^{\alpha(s)} \otimes I_{n}) \overleftarrow{\psi}(s) +  \theta(s) \iota  \| \overleftarrow{\psi}(s) \|^{2}  \notag \\
& \ \ -  \theta(s)\varepsilon   \|\bar{\psi}(s) -x^{*} \|^{2} + 2 \iota \| \overleftarrow{\psi }(s) \| \| \bar{\psi} -x^{*} \|,  \notag \\
& = - \theta(s) \left[
\begin{array}{l}  
\|\overrightarrow{\psi}-\tilde{\psi} \| \\ 
\ \  \|\overleftarrow{\psi} \|
\end{array}
\right]^{T}     
Q
\left[
\begin{array}{l}  
\|\overrightarrow{\psi}-\tilde{\psi} \| \\ 
\ \  \|\overleftarrow{\psi} \|
\end{array}
\right], 
\label{NL6}
\end{align} 
where $Q=\left[
\begin{array}{cc} 
\frac{\varepsilon }{N} & - \frac{\iota}{\sqrt{N}} \\
-\frac{\iota}{\sqrt{N}} & \gamma^{*} \lambda_{min}  - \iota
\end{array}
\right] >0 $  if $ \gamma^{*}  > \frac{1}{\lambda_{min}} (\frac{\iota^{2}}{\varepsilon } +\iota)  $, and 
$ \lambda_{min}=\min \left\lbrace \lambda_{2}(\mathcal{L}^{\sigma(t)}), \ \sigma(t)\in \{1,2,\cdots,m \}  \right\rbrace $. 

\vspace{3pt}
Since $ \mathcal{L}^{\sigma(t)} $ is time-invariant on each time interval  $ [s_{i}, s_{i+1}) $, it is derived that $ W^{'}(s) $ is differentiable on each interval  $ [s_{i}, s_{i+1}) $. Then, by (\ref{PrescribedTimeAlgorithmAdaptiveClosedSwitching}), $ \psi^{i}(s) $ and $ \phi_{ij}(s) $ are bounded for $ [s_{i}, s_{i+1}) $ and further, $  \theta(s) $ is bounded for any $ s\in (0,\infty] $. Thus, there exists a scalar $ \eta>0 $ so that $ \sup_{s_{i} \leq s \leq s_{i+1}, i=0,1,\cdots } |W^{''}(s) | \leq \eta. $  

Then, based on Corollary 1 in \cite{Huang12TAC}, we have $ \lim_{s\rightarrow \infty} W^{'}(s) =0 $. Then, the closed-loop system in (\ref{PrescribedTimeAlgorithmAdaptiveClosedSwitching}) is globally stable for any $\psi^{i}(0) $ and $ \phi_{ij}(0) $ on $ s\in [0,\infty) $. Thus, the estimate states $  \overrightarrow{\psi}(s)-\tilde{\psi} $ and $\overleftarrow{\psi}(s)$ are bounded, and converge the largest invariant set $ \mathcal{M}= \{\overrightarrow{\psi}(s)=\tilde{\psi} \ \text{and} \ \overleftarrow{\psi}(s) = 0_{Nn}\} $.
Further, the gain $ \phi_{ij}(s) $ is bounded and converges to certain finite values. 
\end{proof}		 	 
 
Further, we can show that the proposed NE seeking algorithm  guarantees that the NE can be maintained and the control input signal remains zero over $[T, \infty)$. Further, this signal remains $ \mathcal{C}^{1} $ smooth and uniformly bounded on $[0, \infty)$. The details are similar to the proof of Theorem \ref{The1Boundness}, and are omitted here.

\begin{remark}
Notice that in the absence of prescribed-time and fully distributed requirements, the following corollary is obtained: 

\textbf{Corollary 1:} Under Assumptions 1-3, the following distributed algorithm: $ \dot{\textbf{x}}^{i}=-\mathcal{R}^{T}_{i} \triangledown_{i} J_{i}(\textbf{x}^{i})  + \textbf{e}^{i}, \  \textbf{e}^{i}=-\kappa\sum_{j=1}^{N}a_{ij} (\textbf{e}^{i}-\textbf{e}^{j}) $ enables all players' estimated strategies to exponentially reach a consensus and converge to the NE if $ \kappa >  (\frac{\iota^{2}}{\varepsilon} +\iota)/\lambda_{2}(\hat{\mathcal{L}}) $.

This corollary can cover existing results (e.g., \cite{Ye17TAC} and \cite{Pavel19TAC}) as special cases. It can avoid restrictive graph coupling conditions in \cite{Pavel19TAC} by adding a proportional gain $ \kappa $, and remove the use of high-gain singular perturbation that yields semi-global convergence in \cite{Ye17TAC,Ye19TAC,Pavel19TAC}. Here, it does not require any initial requirements, while certain global graph information is needed.     
\end{remark}

\begin{remark}
The addressed prescribed-time NE seeking is partially inspired by some existing finite-/fixed/appointed-time works on consensus and optimization (e.g.,  \cite{Zhao16AT,Guan12TCS,SongTcyber18,ZhaoAT18, Zan19TAC,Pilloni16VSC,Ren14ACC,Chen18AT,Hu19TCNS,Ding19ACC}).
However, those works encounter certain design limitations such as the dependence of known initial states  \cite{Zhao16AT,Pilloni16VSC,Ren14ACC,Hu19TCNS}, the control parameters \cite{Chen18AT,Ding19ACC}, and unavailable settling time by homogeneity analysis \cite{Guan12TCS}. The appointed-time design in \cite{ZhaoAT18} heavily relies on  motion-planning-based samplings. The most related work in \cite{SongTcyber18} requires two time-varying  functions. Overall, the existing works cannot be directly applied for prescribed-time distributed NE games.   
\end{remark}

\section{Numerical Simulation} 
In this section, numerical examples are presented to verify the effectiveness of the proposed NE seeking designs. 

\vspace{3pt} 
\begin{example}
	(Energy Consumption Game)
\end{example} 

In this example, we consider an energy consumption game of $ N $ players for Heating Ventilation and Air Conditioning (HVAC) system (see \cite{Ye17TAC}), where the cost function of each player $ i $ can be modeled by the following function:
\vspace{-5pt}
\begin{equation}
J_{i}(x_{i},x_{-i})=a_{i}(x_{i}-b_{i})^{2}+ \left(  c\sum_{j=1}^{N}x_{j} +d  \right)  x_{i}, \ i \in \mathcal{V}, \notag  
\end{equation}
where $ a_{i}>0, c>0 $, $ b_{i} $ and $ d $ are constants for $ i \in \mathcal{V} $. It can be verified that Assumptions \ref{AssumptionConvex} and \ref{AssumptionSmooth} are satisfied. Throughout this simulation, let $ a_{i}=1$, $ c=0.1 $, $ d=10 $ for each player. In the following simulation, we investigate the effectiveness of the proposed distributed NE seeking algorithms.

\begin{figure}[!h]
	\centering
	\includegraphics[width=6.0cm,height=2.2cm]{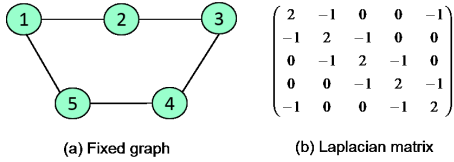}
	\caption{The communication network for a group of five players in the game: (a) fixed graph $ \mathcal{G} $; (b) Laplacian matrix $ \mathcal{L} $.} 
	\label{UndirectedGraph}
\end{figure}

\vspace{-6pt}
\subsection{Prescribed-Time Distributed NE Seeking}
We consider five players ($ N=5 $) in this energy consumption game over an undirected and connected graph as shown in Fig. \ref{UndirectedGraph}. Constants $ b_{i} $ for $ i =1,\cdots,5 $, are set to $ 10$, $15$, $20$, $25 $, and $30$, respectively. By certain calculation based on those parameters, the NE is $ x^{*} =\text{col} (2.0147,6.7766,11.5385,  16. 3004,21.0623)$  \cite{Ye17TAC}. 
The initial states are given by $ x^{i}_{i}(0)=\text{col}(-2,-4,-6, -8, -10)$ and $ x^{i}_{j}(0)=\text{col}(15,10,5,0)$, $\forall i \neq j $, which are not close to $ x^{*} $. The control gain of the prescribed-time distributed NE seeking algorithm in (\ref{PrescribedTimeAlgorithm}) is set as $ \kappa=2 $ and $ c=20 $. 
  
Next, we perform the proposed NE seeking algorithm in (\ref{PrescribedTimeAlgorithm}), and simulation results are shown in Fig. \ref{AlgorithmStatic}. In particular, Fig. \ref{AlgorithmStatic}(a) illustrates all players' estimate strategies on the NE $x^{*}$, while the relative errors of all players' actions $ \|\textbf{x}-\textbf{x}^{*} \| /   \|\textbf{x}^{*}\| $ are depicted in Fig. \ref{AlgorithmStatic}(b). As observed, all players' estimated strategies reach a consensus and converge to the NE within $ T=1.2 $sec, and this NE has been maintained for $ t \geq 1.2 $sec. 

\begin{table}[b!] 
	\centering
	\caption{The performance comparison of different algorithms. }
	\vspace*{-5pt}
	\begin{tabular}{l|l|l|l}
		\hline	\hline
		\diagbox{Algorithms}{Parameters} & $ \kappa $ & c & T (sec)  \\
		\hline
		The   algorithm in \cite{Pavel19TAC}   &  1 & 0 &  21 \\ 
		\hline
		The algorithm in Corollary 1   &  10 & 0 &  13 \\ 
		\hline
		The proposed algorithm in \ref{AlgorithmStatic}    &   2  & 20 &  1.2 \\
		\hline 	
	\end{tabular}
	\label{Table1}
\end{table}

In order to better demonstrate the prescribed-time convergence, a comparison with the following distributed NE seeking algorithm in Corollary 1 and the algorithm in \cite{Pavel19TAC} with $  \kappa=1 $ is provided. 
The simulated results are shown in Figs. \ref{Positionerror}, and the performance comparison is summarized in Table. \ref{Table1}. It can be seen that the proposed NE algorithm illustrates the better property as expected.

\begin{figure}[!t]
\centering
\hspace{0.5em}
\includegraphics[width=8.0cm,height=5.6cm]{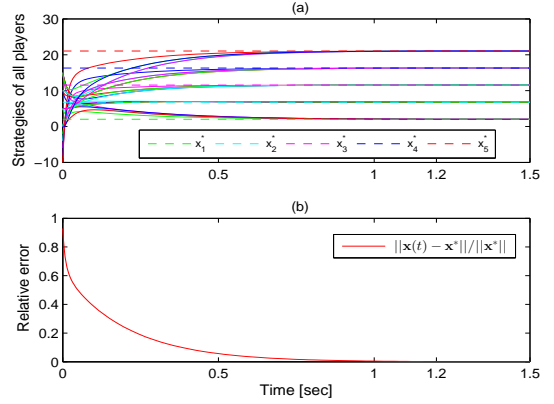}
\caption{Simulated results of the proposed prescribed-time NE seeking algorithm in (\ref{PrescribedTimeAlgorithm}): (a) all players' estimated strategies $x^{i}_{j}(t), i, j\in \mathcal{V}$; and (b) relative errors.} 
\label{AlgorithmStatic}
\end{figure}

%
%

\begin{figure}[t!] 
	\centering
	\hspace*{-0.5em}
	\begin{tabular}{cc}	
		
		\hspace*{-1.5em}		
		{\includegraphics[width=5.0cm,height=5.5cm]{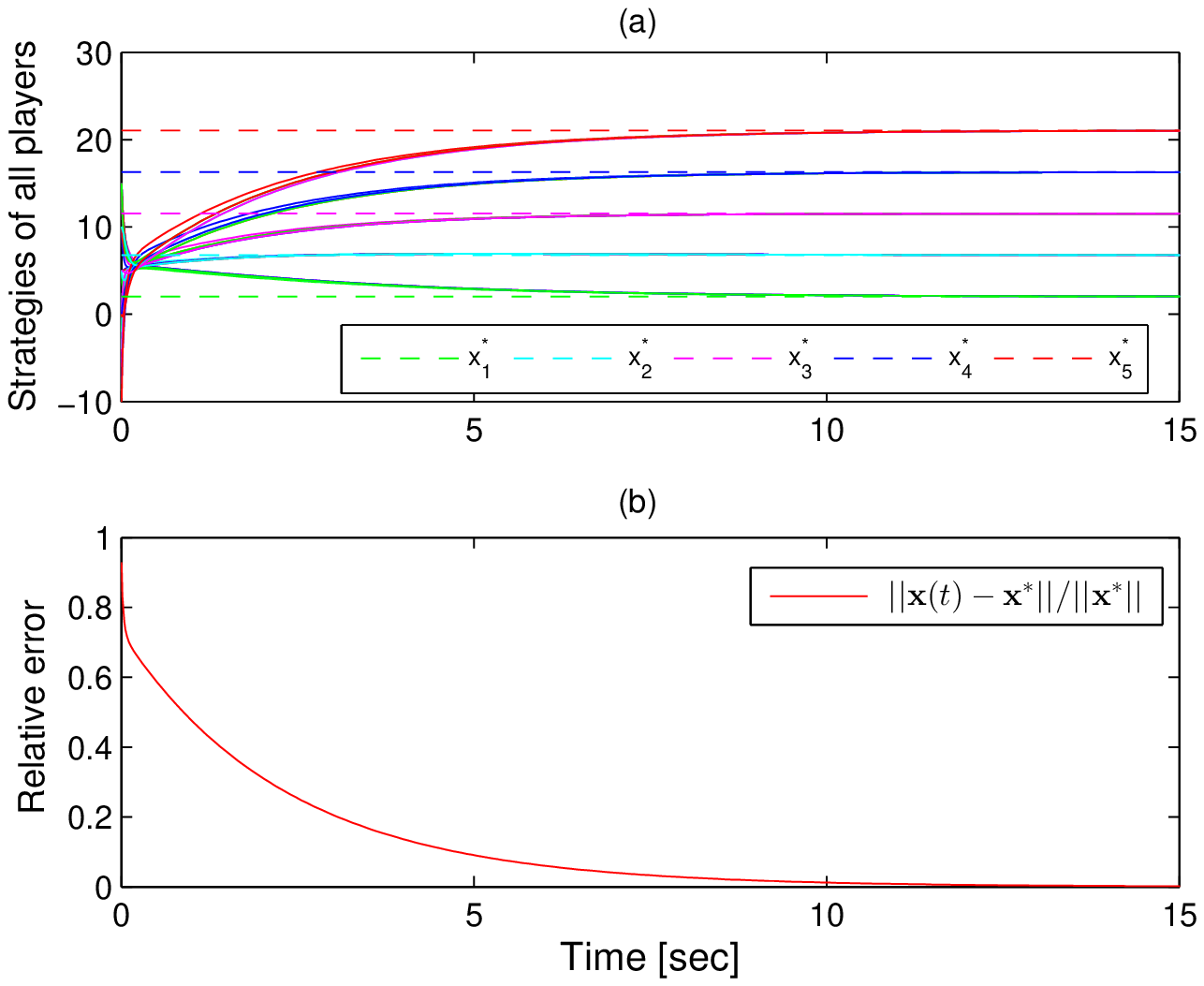}
			\label{L1}}
		
		\hspace*{-1.5em}		
		{\includegraphics[width=5.0cm,height=5.5cm]{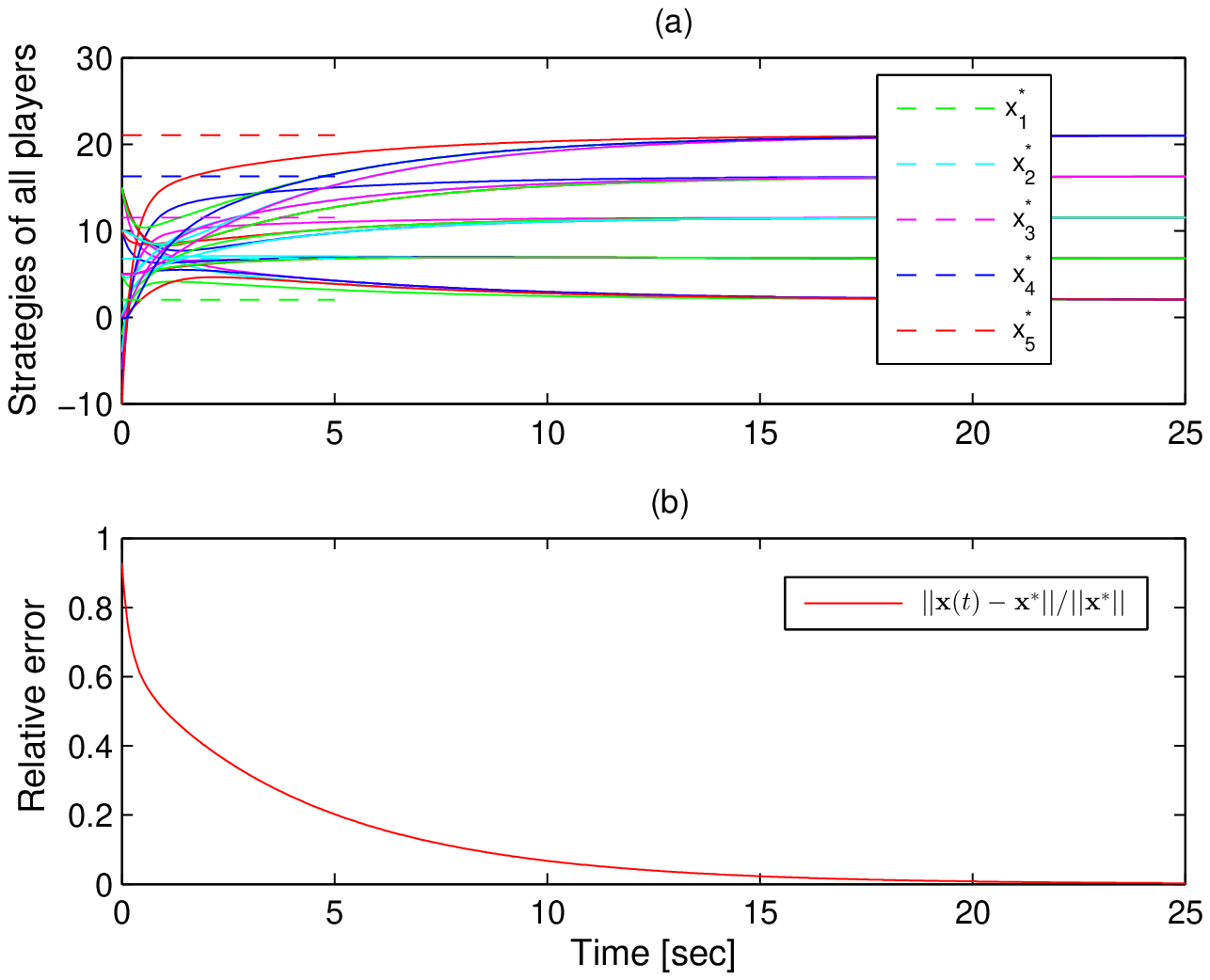}
			\label{L2}}
	\end{tabular}
	\vspace*{-3pt}
	\caption{Simulated results of the algorithms in Corollary1 and \cite{Pavel19TAC}, respectively.}
	\label{Positionerror}
\end{figure}

\subsection{Prescribed-Time And Fully Distributed NE Seeking}
In this part, the proposed prescribed-time and fully distributed algorithm in \ref{AlgorithmAdaptive} is performed with the same simulation setting in the subsection V-A. The initial states for the dynamic parameter are set as $\kappa_{ij}(0)=0.5 $. The simulation results are shown in Figs. \ref{AlgorithmAdaptive}-\ref{DynamicGain}, where the plots of all players' strategies and relative errors are shown in Fig. \ref{AlgorithmAdaptive} and the trajectories of the dynamic gain are depicted in Fig. \ref{DynamicGain}, which converge to certain  constants. It follows from figures that players' actions reach a consensus and converge to the NE in a prescribed-time and fully distributed manner.

\begin{figure}[!t]
\centering
\includegraphics[width=8.0cm,height=5.5cm]{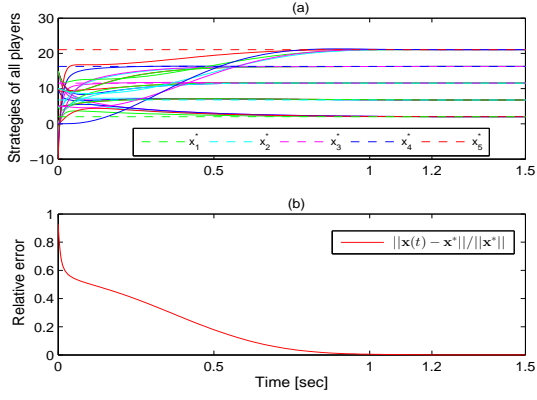}
\caption{Simulated results of the proposed algorithm in (\ref{PrescribedTimeAlgorithmAdaptive})
: (a) all players' estimated strategies $x^{i}_{j}(t), i, j\in \mathcal{V}$; and (b) relative errors of all players' actions.
} 
\label{AlgorithmAdaptive}
\end{figure}

\begin{figure}[!t]
\centering
\includegraphics[width=8.0cm,height=5.5cm]{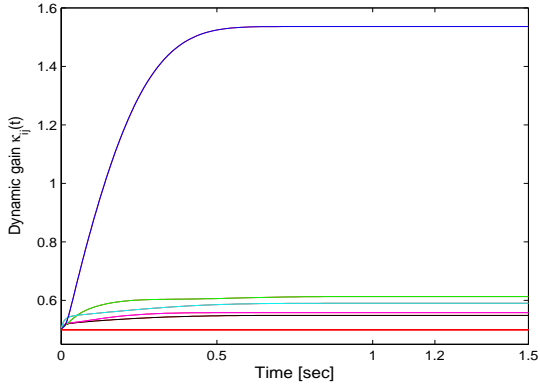}
\caption{The plot of all players' dynamic gain $\kappa_{ij}(t), i, j\in \mathcal{V}$ generated by (\ref{PrescribedTimeAlgorithmAdaptive}).} 
\label{DynamicGain}
\end{figure}

\vspace{-6pt}
\subsection{Prescribed-Time And Fully Distributed NE Seeking over Jointly Switching Topologies}
In this simulation, the proposed algorithm over jointly switching network topology $\mathcal{G}_{\sigma(t)}$ will be conducted. Specifically, we  consider the topologies dictated by the following switching signal: 
\vspace{-3pt}
\begin{equation}
\sigma(t)=\left\{ 
\begin{array}{c}
\hspace{-3.5em} 1, \ \text{if} \ s \varGamma \leq t< (s+0.25) \varGamma; \\ 
2, \ \text{if} \ (s+0.25)\varGamma \leq t< (s+0.5) \varGamma; \\
3, \ \text{if} \ (s+0.5)\varGamma \leq t< (s+0.75) \varGamma; \\
\hspace{-0.9em} 4, \ \text{if} \ (s+0.75)\varGamma \leq t< (s+1) \varGamma, 
\end{array}
\right. 
\label{SwitchingSignal}
\end{equation}
where $ \varGamma=0.4 $sec., and $ s=0,1,2,\cdots $. The signal $ \sigma(t) $ defines fixed graphs $\mathcal{G}_{i}$, $ i=1,2,3,4 $ as shown in Fig. \ref{SwitchingGraphs}. As can be seen, Assumption \ref{AssumptionGraphSwitching} is satisfied even though $\mathcal{G}_{i}$ is disconnected at $ t\geq 0 $.

Next, we perform the proposed algorithm in the form of (\ref{PrescribedTimeAlgorithmAdaptiveSwitching}) to accommodate switching graphs with the same simulation setting in the subsection V-B. The simulation results are shown in Figs. \ref{AlgorithmAdaptiveSwitching}-\ref{AlgorithmAdaptiveSwitchingK}. It can be seen that all players' actions can reach a  consensus and converge to the NE in a prescribed-time and fully distributed manner over jointly switching graphs.

\begin{figure}[!h]
\centering
\includegraphics[width=8.0cm,height=6.45cm]{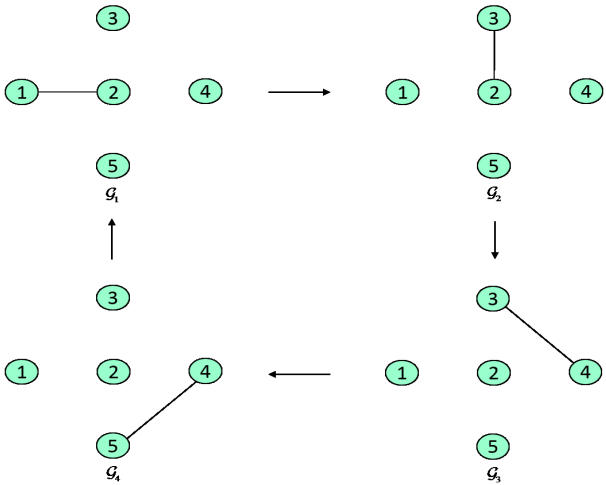}
\caption{The jointly switching topologies $\mathcal{G}_{i}$ with $ \mathcal{P}={1,2,3,4} $.} 
\label{SwitchingGraphs}
\end{figure} 

\begin{figure}[!t]
	\centering
	\includegraphics[width=8.0cm,height=5.5cm]{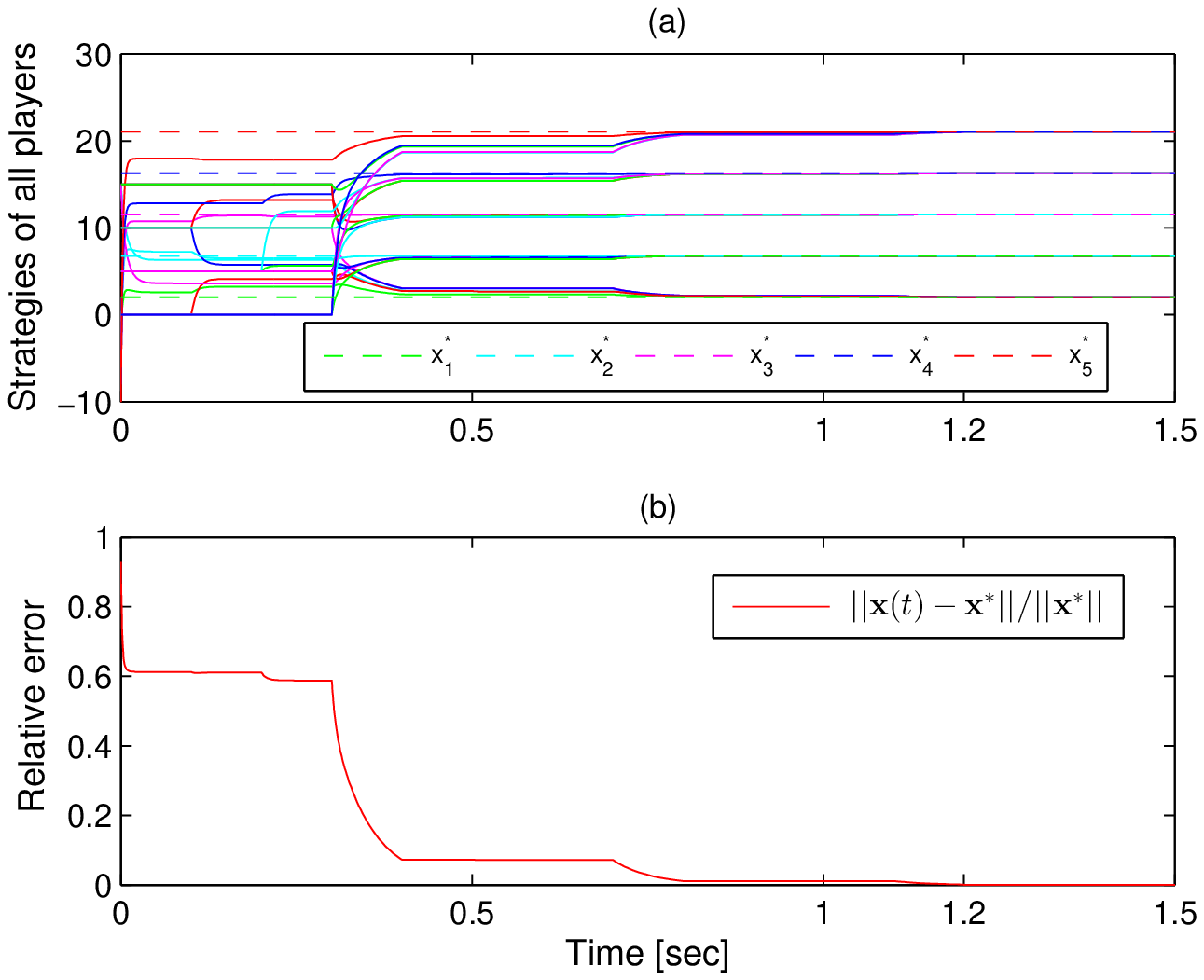}
	\caption{Simulated results of the proposed algorithm in (\ref{PrescribedTimeAlgorithmAdaptiveSwitching}): (a) all players' estimated strategies $x^{i}_{j}(t), i, j\in \mathcal{V}$; and (b) relative errors of all players' actions.} 
	\label{AlgorithmAdaptiveSwitching}
\end{figure} 

\begin{figure}[!t]
	\centering
	\includegraphics[width=8.0cm,height=5.4cm]{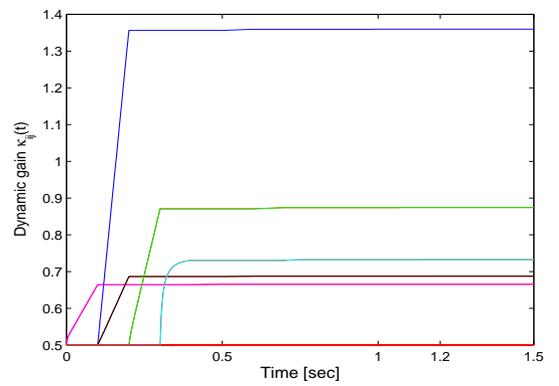}
	\caption{The plot of all players' dynamic gain $\kappa_{ij}(t), i, j\in \mathcal{V}$ generated by (\ref{PrescribedTimeAlgorithmAdaptiveSwitching}).} 
	\label{AlgorithmAdaptiveSwitchingK}
\end{figure}

\begin{example}
	(General Non-Quadratic Game)
\end{example}

In this example, we investigate a more general non-quadratic noncooperative game, in which the objective functions for each player $ i $ are described by 
\vspace{-3pt}
\begin{align}
J_{1}(x_{1},x_{-1})&= \frac{x^{2}_{1}}{2}+x_{1}\sum^{5}_{j=2}x_{j},   J_{2}(x_{2},x_{-2})= \frac{e^{\frac{x_{2}}{2}}}{2}+x_{2}x_{4}, \notag \\
J_{3}(x_{3},x_{-3})&=\frac{x^{2}_{3}}{2}+x_{1}^{3}, \
J_{4}(x_{4},x_{-4})= \ln(e^{x_{4}})+x^{2}_{4}+x^{3}_{3}, \notag \\
J_{5}(x_{5},x_{-5})&= x^{2}_{5}-5x_{5}+x^{3}_{1}x_{2} +x_{3}x^{4}_{4}, \ i=1,\cdots,5.
\end{align} 

\vspace{2pt}
Based on calculations, the NE is $ x^{*} =\text{col} (-4.6589,4.1589,0,\\  -2,2.5)$. The switching communication topologies are depicted in Fig. \ref{SwitchingGraphs}. Without loss of generality, we only perform the proposed prescribed-time and fully distributed NE algorithm in (\ref{PrescribedTimeAlgorithmAdaptiveSwitching}) for this non-quadratic game. The same simulation settings as those in the subsection V-C are considered. 

The simulation result is depicted in Fig. \ref{NonQuadraticGame}, where all players' strategies and relative errors are provided in Fig. \ref{NonQuadraticGame} (a) and Fig. \ref{NonQuadraticGame} (b), respectively. As we can see, under the proposed algorithm, all players' estimates  reach a consensus and converge to the NE of this non-quadratic noncooperative game in the  prescribed-time and fully distributed manner over switching graphs.

\begin{figure}[!h]
\centering
\hspace{-0.5em}
\includegraphics[width=9.0cm,height=8.5cm]{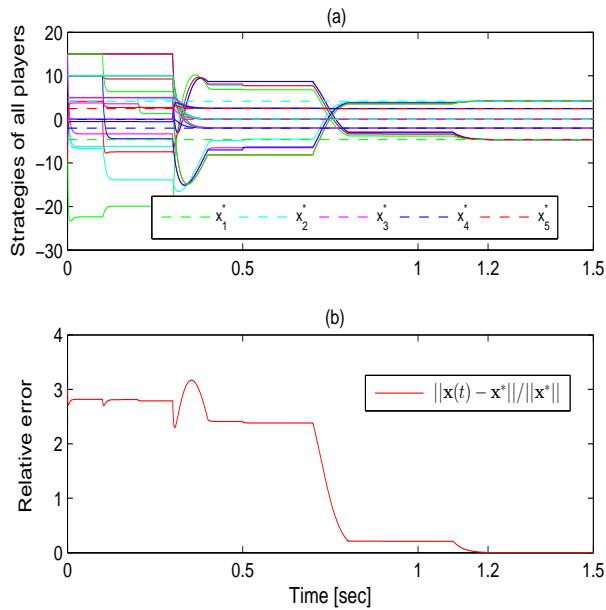}
\caption{Simulated results of the proposed algorithm in (\ref{PrescribedTimeAlgorithmAdaptiveSwitching}): (a) all players' estimated strategies $x^{i}_{j}(t), i, j\in \mathcal{V}$; and (b) relative errors of all players' actions.} 
\label{NonQuadraticGame}
\end{figure}

\section{Conclusion} 
In this paper, prescribed-time fully distributed algorithms have been presented for NE seeking of non-cooperative games, where players' strategies are updated through a communication graph. The significant feature of proposed algorithms is that the global convergence of the NE is achieved in a prescribed-time and fully distributed manner. That is, the convergence time is user-defined according to task requirements and it is independent of any initial conditions and design parameters. Moreover, the proposed fully distributed algorithm does not require any global information on the algebraic connectivity of graphs, the Lipschitz and monotone constants of pseudo-gradients, and the number of players. Lastly, we have extended this algorithm to accommodate jointly switching graphs. The effectiveness of the developed approach has been illustrated by the numerical examples. Further work may consider distributed NE seeking problems for noncooperative games with coupled equality and/or inequality constraints.


\end{document}